%% file: main.tex
\documentclass{article}

\usepackage[preprint,nonatbib]{neurips_2019}

\usepackage{booktabs} 
\usepackage[ruled]{algorithm2e} 
\usepackage{amsmath,amsfonts,amssymb,graphicx,indentfirst,amsthm,bm}
\usepackage{subcaption}
\usepackage[utf8]{inputenc}
\usepackage[normalem]{ulem}
\usepackage{dsfont}
\newtheorem{definition}{Definition}
\newtheorem*{remark}{Remark}
\newtheorem{innercustomthm}{Theorem}
\newenvironment{customthm}[1]
  {\renewcommand\theinnercustomthm{#1}\innercustomthm}
  {\endinnercustomthm}
  
  \newtheorem{innercustomlemma}{Lemma}
\newenvironment{customlemma}[1]
  {\renewcommand\theinnercustomlemma{#1}\innercustomlemma}
  {\endinnercustomlemma}

\newcommand{\newThreshNotation}{\tilde{\beta}^{(\epsilon)}_r}

\newcommand{\newThreshNotationEpsn}{\tilde{\beta}^{(\epsilon_n)}_r}

\newcommand{\tbeta}{{\tilde{\beta}_{x_1}}}
\newcommand{\dtbeta}{{\tilde{\beta}_{x_0, x_1}}}
\newcommand{\obeta}{{\beta^*_{R}}}
\DeclareMathOperator*{\argmax}{argmax}

\usepackage{ourCommonMacros}
\usepackage{xcolor}
\usepackage{tikz,tkz-tab}

\usepackage[utf8]{inputenc} 
\usepackage[T1]{fontenc}    
\usepackage{hyperref}       
\usepackage{url}            
\usepackage{booktabs}       
\usepackage{amsfonts}       
\usepackage{nicefrac}       
\usepackage{microtype}      

\title{Robust Stackelberg buyers in repeated auctions}

  \author{
  Cl\'ement Calauz\`enes\\
  Criteo AI Lab\\
 \texttt{c.calauzenes@criteo.com} 
  \And
  Thomas Nedelec \\
 Criteo AI Lab, ENS Paris Saclay \\
 \texttt{nedelec@cmla.ens-cachan.fr}
 \And
Vianney Perchet \\
ENS Paris Saclay, Criteo AI Lab,  \\
 \texttt{vianney.perchet@normalesup.org}
 \And
 Noureddine El Karoui \\
 UC, Berkeley, Criteo AI Lab \\
 \texttt{nkaroui@berkeley.edu}
  }

\begin{document}

\maketitle

\begin{abstract}
We consider the practical and classical setting where the seller is using an exploration stage to learn the value distributions of the bidders before running a revenue-maximizing auction in a exploitation phase. In this two-stage process, we exhibit practical, simple and robust strategies with large utility uplifts for the bidders. We quantify precisely the seller revenue against non-discounted buyers, complementing recent studies that had focused on impatient/heavily discounted buyers. We also prove the robustness of these shading strategies to sample approximation error of the seller, to bidder's approximation error of the competition and to possible change of the mechanisms.   
\end{abstract}

\section*{Introduction}
Repeated auctions play an important role in modern economics as they are widely used in practice to sell e.g.  electrical power or digital goods such as ad placements on big online platforms. Understanding the precise interactions between the buyers and the sellers in these auctions is key to assess the balance of power on big online platforms. In practice, most of the online auctioneers are using tools at the intersection of classical auction literature \cite{Myerson81} and statistical learning theory. They take advantage of the enormous amount of data they gather on the behavior of the buyers to learn optimal auctions and maximize revenue for the platforms. 

Myerson showed how to design an incentive-compatible revenue-maximizing auction once the seller knows the value distributions of the buyers \cite{Myerson81}. If the seller has perfect knowledge of these distributions, she can define the allocation and payment rules maximizing her expected revenue.

How does the seller learn these value distributions in practice to create her optimal revenue-maximizing auction ? A very rich line of work \cite{OstSch11,cole2014sample,medina2014learning,morgenstern2015pseudo,dhangwatnotai2015revenue,huang2018making} has assumed that the seller has access to a finite sample of bidders' valuations, coming from past bids in truthful auctions. The game they consider is divided in two stages. The first round consists in several truthful auctions such as a second price auction without reserve or with random reserve where the bidders bid truthfully providing the seller with draws from their value distribution. Then, if the bidders are not strategic in the first stage, the seller can learn an approximation of the revenue-maximizing auction based on these samples and run a revenue-maximizing auction in the second stage. 

A first way to prevent the bidders from being strategic is to adapt the mechanism (e.g. the reserve price shown) based on the competition of a bidder rather than based on the bidder themselves \cite{ashlagi2016sequential,kanoria2017dynamic,epasto2018incentive}. A limitation of this type of approach is the need to rely on (partial) symmetry of the bidders: any bidder is assumed to have some competitors with (almost) the same value distribution as them. In particular and as noticed in \cite{epasto2018incentive}, it cannot handle the existence of any dominant buyer, i.e., a buyer with higher values than the other bidders. This is therefore a somewhat limited setting as revenue-optimizing mechanisms are especially needed when buyers are heterogenous.

On big online platforms, the same dominant bidders are actually repeatedly interacting with a seller,  billions of times a day in the case of online advertising. This setting has been considered from either the seller's point of view in  \cite{amin2013learning,amin2014repeated, CesGenMan13, mohri2015revenue,golrezaei2018dynamic} or the bidders' standpoint \cite{kanoria2017dynamic, tang2016manipulate, nedelec2018thresholding, nedelec2019learning}. All these works can be seen as a special instance of a Stackelberg game \cite{differential_games_book_CUP_2000} where the bidders know the rules of the mechanism and have the choice of the bid distribution they will disclose to the seller. The main takeaway from these works is that if one of the seller or bidders is extremely dominant in term of \emph{patience} -- i.e. longer time horizon of revenue optimization -- then this player can get the best revenue/utility to hope for: the payment of a \emph{revenue-maximizing auction} if the seller is strongly dominant, the utility of a \emph{2$^\text{nd}$-price auction with no reserve price} if the bidders are strongly dominant and are all strategic. Yet, these payments have been exhibited in extreme asymptotic cases. In the line of work following from \cite{amin2013learning, mohri2015revenue}, the bidder is assumed to be finitely patient while the seller is infinitely patient. In this strongly unbalanced setting, the seller is able to begin with exploration stages long enough to force the bidder to be truthful, allowing the seller to play the revenue-maximizing auction in the (longer) exploitation phase. On the contrary, if the bidders are infinitely patient and the seller has to update the mechanism in finite time, \cite{tang2016manipulate, nedelec2018thresholding} exhibited optimal strategies that bidders can enforce. The remaining crucial question is: what happens in between these extreme cases, in more realistic conditions?

Our work aims at providing answers to this question by studying the robustness of the bidders' strategic behavior to more realistic conditions:
\begin{enumerate}
\item Are these strategic behaviors robust to strategic sellers implementing two-stage process (exploration/exploitation) such as the ones described in \cite{amin2013learning, golrezaei2018dynamic}? How is the value shared between seller and bidders in non-asymptotic ``patience" conditions?
\item Are the possible strategic behaviors robust to the seller optimizing the mechanism with a finite number of samples?
\item How does the lack of knowledge of the competition (i.e. other bidders's value distributions) negatively impact the strategic bidder's utility?
\end{enumerate}

\section{Framework and contributions}
As it is classical in auction theory \cite{krishna2009auction}, we assume the valuation of a bidder $v \in \mathbb{R}$ is drawn from a specific distribution $F$ (the distribution can be different from one bidder to the other);  a bidding strategy is a mapping $\beta$ from $\mathbb{R}_+$ into  $\mathbb{R}_+$ that provides the actual bid $B=\beta(v)$ when the value is $v$. As a consequence, the distribution of bids $F_{B}$ is  the push-forward of $F$ by $\beta$. We assume from now on that the support of $F$ is $[0, b) \subseteq [0, +\infty]$. If $b = +\infty$, we additionally assume that $F$ verifies $1-F(x) = o(x^{-1})$ to avoid problems with the definition of the optimal reserve price. Unless otherwise noted $F$ is assumed to be regular.

The seller has an estimate of the bidders' bid distribution, typically coming from access to past bids. Then, she uses an approximation of a revenue-maximizing auction based on her estimate of the bid distributions. We will consider that this is a lazy 2$^\text{nd}$-price auction \cite{paes2016field} whose reserve price is optimized to maximize the \emph{seller's revenue}. For the lazy 2$^\text{nd}$-price auction, the optimal reserve price happens to be the \emph{monopoly price} $r^*_{\beta}$ which maximizes the \emph{monopoly revenue} 
$$
R_{B}(r) = r(1-F_{B}(r))
$$ extracted by the seller on the current bidder. We can formalize this as a two-step process \cite{amin2013learning,cole2014sample}: 
\paragraph{General two-step process with commitment:} A general two-step process with commitment is a tuple of the form $P= (G, H,r, F)$ that is defined as follows: 
\begin{description}
\item[1 -- exploration] The seller runs a lazy 2$^\text{nd}$-price auction. The current bidder faces a competition $G$. In this step, the potential randomized or deterministic reserve price is denoted by the distribution $H$.
\item[2 -- exploitation] The seller runs a lazy 2$^\text{nd}$-price auction with reserve price $r$, possibly personalized. The current bidder faces the same competition $G$ as in the first step.
\end{description}

To make explicit the dependency on the bidder's strategy $\beta$ on the design of the auction of the second step when optimizing the reserve price based on the observation of the first period, we use the notation $P^{\beta} = (G, H, r^*_{\beta}, F)$.
The tradeoff between exploration and exploitation from the seller standpoint was introduced in \cite{amin2013learning} and refined in \cite{mohri2015revenue, golrezaei2018dynamic}. They introduce a parameter $\alpha$, $0\leq \alpha<1$, to define this trade-off, assuming the ratio of length between the first and the second stage is equal to $\alpha/1-\alpha$. \cite{amin2013learning} proved that if bidders are non-discounted buyers, there must exist a good strategy for them in this mechanism, forcing the seller to suffer a regret linear in the number of auctions. We are interested in solving the Stackelberg game faced by bidders when they know the seller is using this classical mechanism to learn their value distribution. Bidders are the leaders in this framework since they know the mechanism used by the seller and can choose their strategy accordingly.

We assume that the strategic bidder commits to the same strategy $\beta$ (inducing the same distribution of bids $F_{B}$) in both phases. This assumption accounts for the fact that in practice sellers regularly update their priors based on recent past bids, forcing the bidders to commit in the long-term to the bid distribution they want the seller to use as prior. Otherwise, the seller might discover new aspects of the buyer's bid distribution and change their mechanism accordingly more than once as time evolves.

In this framework, the objective of the bidder is to choose a strategy $\beta$ to maximize a weighted sum of her utility in both phases:
\begin{align}
U_\alpha(\beta) = \alpha U_1(\beta) + (1-\alpha) U_2(\beta)
\end{align}
where $U_i$ is the expected utility of the bidder in stage $i$. The $\alpha \in [0,1]$ quantifies the length of each stage \cite{amin2013learning}. More precisely, $U_1$ is the expected utility of a lazy second price without reserve/with random reserve and $U_2$ is the expected utility of a lazy second price with monopoly reserves corresponding to $F_B$.

\cite{tang2016manipulate, nedelec2018thresholding} exhibited strategies for the case where the bidders only optimize for the utility of the second phase, i.e. $\alpha=0$. However, in practice, sellers use an exploration stage to learn the value distribution. Taking care of this exploration phase is of great importance for the bidders since optimizing only the second stage can lead to large loss of utility during the first stage. 

In Section \ref{sec:two-stage-process}, we extend the approach of \cite{tang2016manipulate, nedelec2018thresholding} to the two-stage game showing that the strategies which are optimal when the bidder is only optimizing her utility in the second stage are suboptimal when she takes into consideration the exploration stage of the seller. We find in Section \ref{sec:welfareSharing} the optimal strategy in this framework and study the behavior of this strategy as a function of the length of the exploration stage. To study whether the strategies are robust to the presence of other strategic bidders in the game, we also prove the existence of a Nash equilibrium under certain conditions and we compute the utility of the bidders and the revenue of the seller at this Nash equilibrium when it exists. 

Unlike \cite{amin2013learning, golrezaei2018dynamic}, we consider the objective function of the bidders to be the expected utility of the two stages instead of the utility computed on a finite number of auctions. Indeed, as there exists a high number of repeated auctions, billions a day in the case of online advertising, optimizing directly the expected utility makes sense from a bidder's standpoint. We extend our results to finite sample sizes in Section \ref{sec:empirical_estimation}. We finally show the strategies are robust to an estimation of the bidders' competition in Section \ref{sec:robust_competition} and to a change of mechanism in Section \ref{sec:change_mechanism}. It shows their robustness to most of the difficulties bidders practically face and provides concrete solutions to solve these problems.

\section{Understanding the strategic reaction of the players}
\label{sec:two-stage-process}
In order to get a good understanding of how the welfare is shared between seller and bidders when $\alpha$ moves from $0$ to $1$, we need to study the strategies of the different players. First, on the seller's side, she chooses the distribution of reserve prices $H$ faced by the strategic bidder in the first phase. At this point, we assume the seller does not have any knowledge about the bidders (except the support of the value distribution). Hence, she mostly has two choices for $H$: either no reserve price or a uniform distribution \cite{amin2014repeated}. In the second stage, the seller is assumed welfare-benevolent, i.e. if she has the choice between two reserve prices equivalent in terms of payment, she chooses the lowest one. Hence, in the second stage, the seller chooses $r^*_\beta = \inf \argmax_r R_B(r)$. 

From a bidder's standpoint, \cite[Th. 6.2]{tang2016manipulate} exhibits the best response for the extreme case $\alpha = 0$. Unfortunately, this best response induces a complicated optimization problem for the seller: $R_{B}$ is non-convex with several local optima and the global optimum is reached at a discontinuity of $R_{B}$ (see Fig. \ref{fig:optimal_revenues}). This is especially problematic if sellers are known to optimize reserve prices conditionally on some available context using parametric models such as Deep Neural Networks \cite{dutting2017optimal} whose fit is optimized via first-order optimization and hence would regularly fail to find the global optimum. This is undesirable for the bidders: in any Stackelberg game, the leader's advantage comes from being able to predict the follower's strategy.

To address this issue, \cite{tang2016manipulate,nedelec2018thresholding} proposed a \emph{thresholding strategy} that is the best response in the restricted class of strategies that ensure $R_{B}$ to be concave as long as $R_{X}$ is so, when $\alpha = 0$. We show this strategy can be extended to the two-stage process with a slight modification. The requirement is to ensure $R_{B_i}$ to be \emph{quasi-}concave as long as $R_{X}$ is so. Then, we have the following result.
\begin{theorem}
Given a two-step process $(G, H, r_{\beta}^*, F)$ with $r^*_\beta = \inf \argmax_r R_B(r)$ and such that $F$ is quasi-regular\footnote{We say a distribution is quasi-regular if the associated revenue curve $R_{X_i}$ is quasi-concave}, $\forall \alpha\in[0,1], \,\exists 0 \leq x_0 \leq x_1$ such that the best quasi-regular response (maximizing $U_\alpha(\beta)$ with $F_B$ regular) is 
\vspace{-0.5em}
\begin{align}\label{eq:defDoubleThresh}
\dtbeta(x)
=
\indicator{x \leq x_0}x
+
\indicator{x_0 < x \leq x_1}\frac{R}{1-F(x)}
+
\indicator{x > x_1}x
~~~~~\text{where}~~~~~ R = x_1 (1 - F(x_1))
\end{align}
Moreover, we have $x_1 = \sup \{x: x(1-F(x)) \geq R\}$ and $x_0 \leq \overline{x}_0 = \inf \{x: x(1-F(x)) \geq R\}$.
\label{dfn:thresholded_beta}
\end{theorem}
The proof is in Appendix \ref{ssec:best_quasi_regular_response}.
\begin{remark}
The thresholding strategy from \cite{tang2016manipulate,nedelec2018thresholding} is a special case of $\dtbeta$ in the case when $x_0 = 0$ and we will denote it $\tbeta$ for simplicity. As we will see in Section \ref{sec:optimal_strategies}, when $\alpha$ is small enough (but not 0 yet), $x_0 = 0$ is optimal. It thus makes sense to study both these strategies.
\end{remark}
This class of strategies solves the main drawback of the best response from \cite{tang2016manipulate}. If the value distribution $F_i$ is quasi-regular, then the bid distribution $F_{B}$ is also quasi-regular, offering a quasi-concave maximization problem (thus ``predictable" for the buyer) to the seller.
Figure \ref{fig:optimal_revenues} illustrates the different strategies as well as the corresponding optimization problems and the virtual values associated to the push-forward bid distributions. Understanding the strategic answer of the bidders helps avoiding worst-case scenario reasoning when studying the revenue of the sellers against strategic bidders. We now quantify precisely how the welfare is shared between strategic bidders and seller.
\begin{figure}[tbp]
  \centering
  \includegraphics[width=0.41\textwidth]{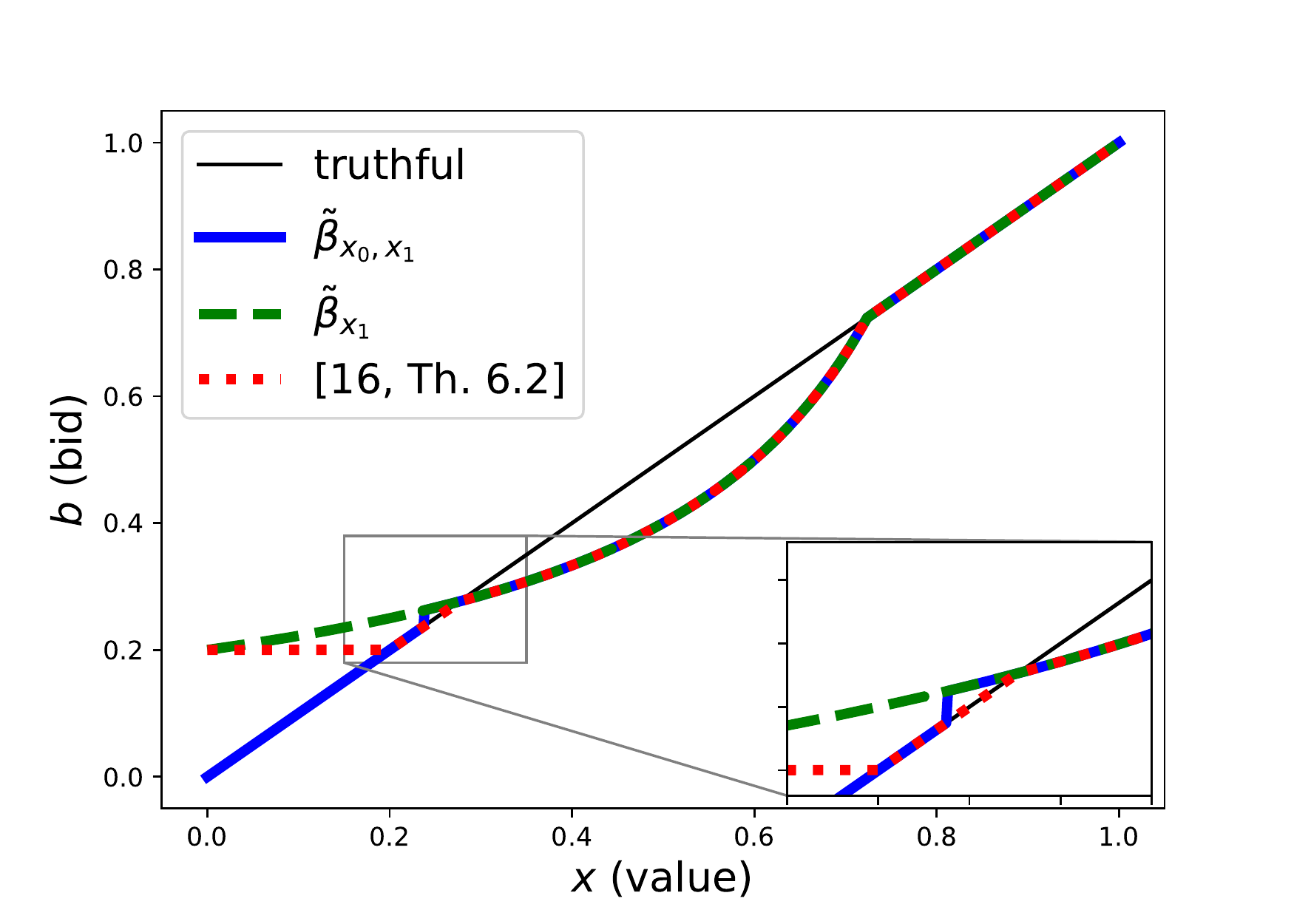}
  \includegraphics[width=0.41\textwidth]{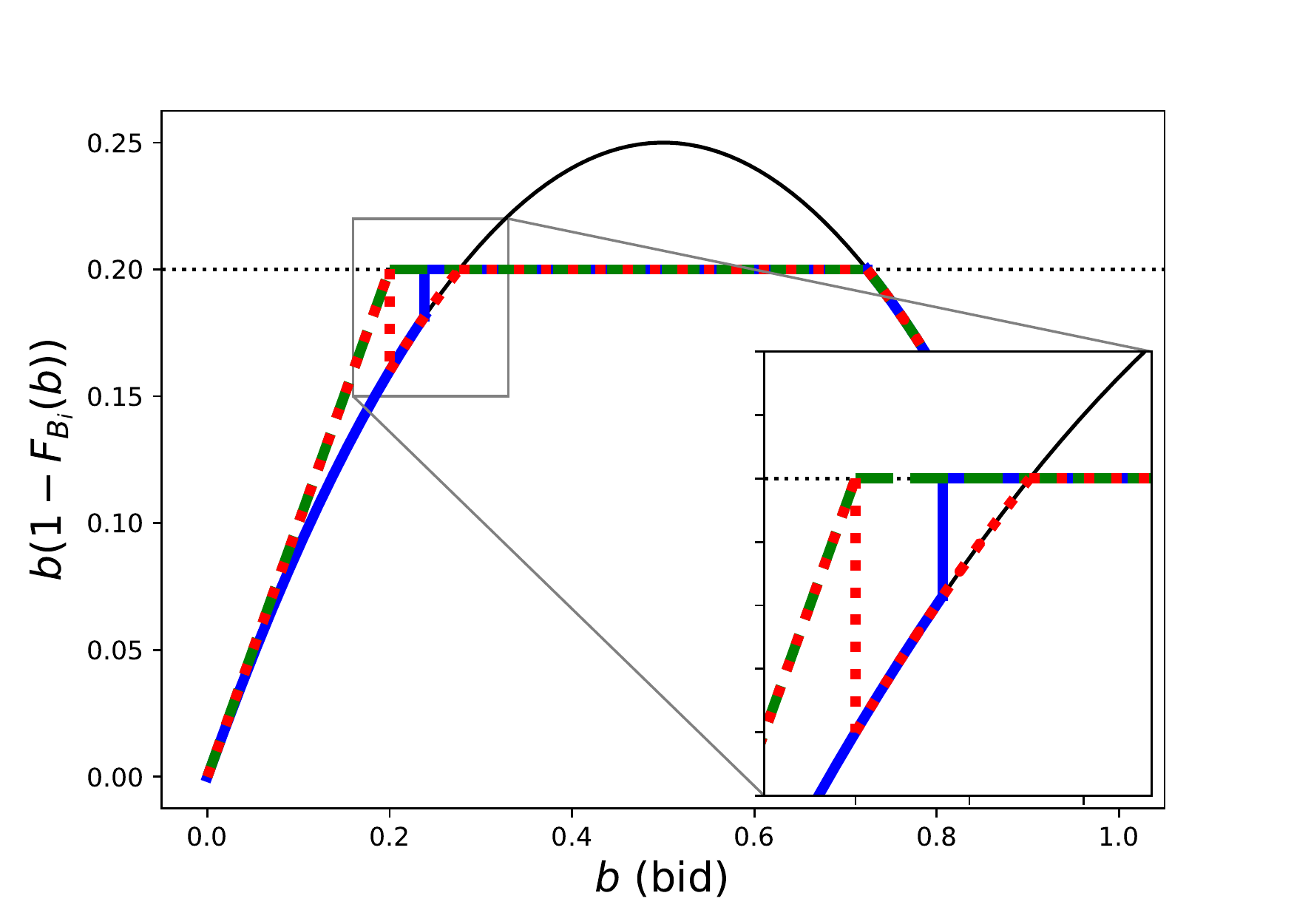}
  \caption{Left: illustration of strategies. Right: function optimized by the seller $b(1-F_{B}(b))$ depending on the strategy for values distribution $F$ being $\mathcal{U}(0,1)$. $R$ is set at 0.2, $x_1$ is set such that $x_1(1-F(x_1)) = R$ (right root) and $x_0 \simeq 0.27$ for $\dtbeta$.}
  \label{fig:optimal_revenues}
\end{figure}

\section{Welfare sharing between seller and buyers}\label{sec:welfareSharing}
This section presents how the welfare is shared in the two-stage process with strategic bidders. First, we show how to numerically compute the best response and illustrate the variation of utility and payment with $\alpha$. After remarking and showing that the \emph{thresholded} strategy \cite{nedelec2018thresholding} is optimal for greater values of $\alpha$ than 0, we focus on this strategy and show the existence of a Nash equilibrium. It proves that even in the case of multiple strategic bidders, there exists strategies that enable bidders to take advantage of the learning stage of the seller.
\subsection{Welfare sharing with best quasi-regular response}
\label{sec:optimal_strategies}
We first show how to obtain numerically the best response $\dtbeta$ by solving the maximization problem of $U_\alpha$ restricted to the class of strategies described in Th. \ref{dfn:thresholded_beta}. This allows us to show how the bidder's utility and payment vary with $\alpha$.
\begin{theorem}
Given a two-stage process $(G, H, r_{\beta}^*, F)$ with $r^*_\beta = \inf \argmax_r R_B(r)$, the best response is of the form of $\tilde{\beta}_{x_0^*, x_1^*}$ and the utility $U_\alpha(\dtbeta)$ has the following derivatives:
\begin{align}
    \frac{\partial U_\alpha(\dtbeta\!)}{\partial x_0} \!=\! \alpha(1-F(x_0))G(x_0)H(x_0)
    -
    x_0 f(x_0) G_\alpha\!\!\left(\frac{x_1(1-F(x_1))}{(1-F(x_0))}\right)~~~~~~~~~~~\\
    \frac{\partial U_\alpha(\dtbeta\!)}{\partial x_1} 
    \!= \!
    f(x_1)\psi(x_1)
    \!\left(\!G_\alpha(x_1) - \lE_X\!\!\left(\!\indicator{x_0 \leq X \leq x_1}\frac{X}{(1-F(X))}g_\alpha\!\!\left(\!\frac{x_1(1-F(x_1))}{(1-F(X))}\!\right)\!\!\right)\!\!\right)
\end{align}
where $G_\alpha(x) = \alpha G(x)H(x) + (1-\alpha)G(x)$ and $\psi(x)=x-\frac{1-F(x)}{f(x)}$.
\label{thm:double_thresholded_derivatives}
\end{theorem}
The proof is in Appendix \ref{ssec:welfare_sharing_best_response}. We now use these results to compute numerically the best response in the following two-stage process: $P_1 = (G\!=\!U_{[0,1]}, H\!=\!0, r^*_{\tilde{\beta}_{x_0^*, x_1^*}}, F\!=\!U_{[0,1]})$ (uniformly distributed reserve in  stage 1)

\begin{figure}[tbp]
  \centering
 
    \includegraphics[width=0.32\textwidth]{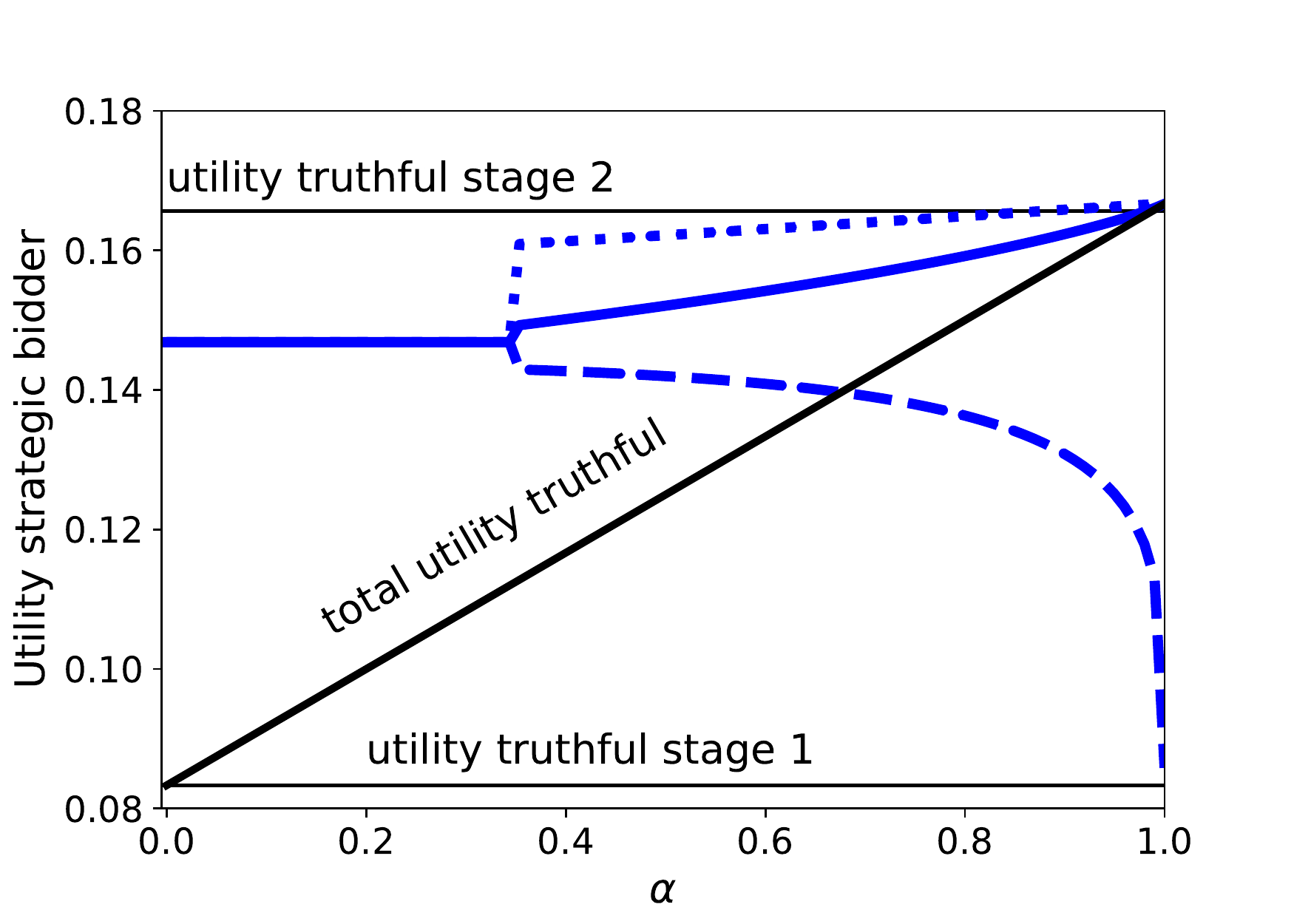}
     \includegraphics[width=0.32\textwidth]{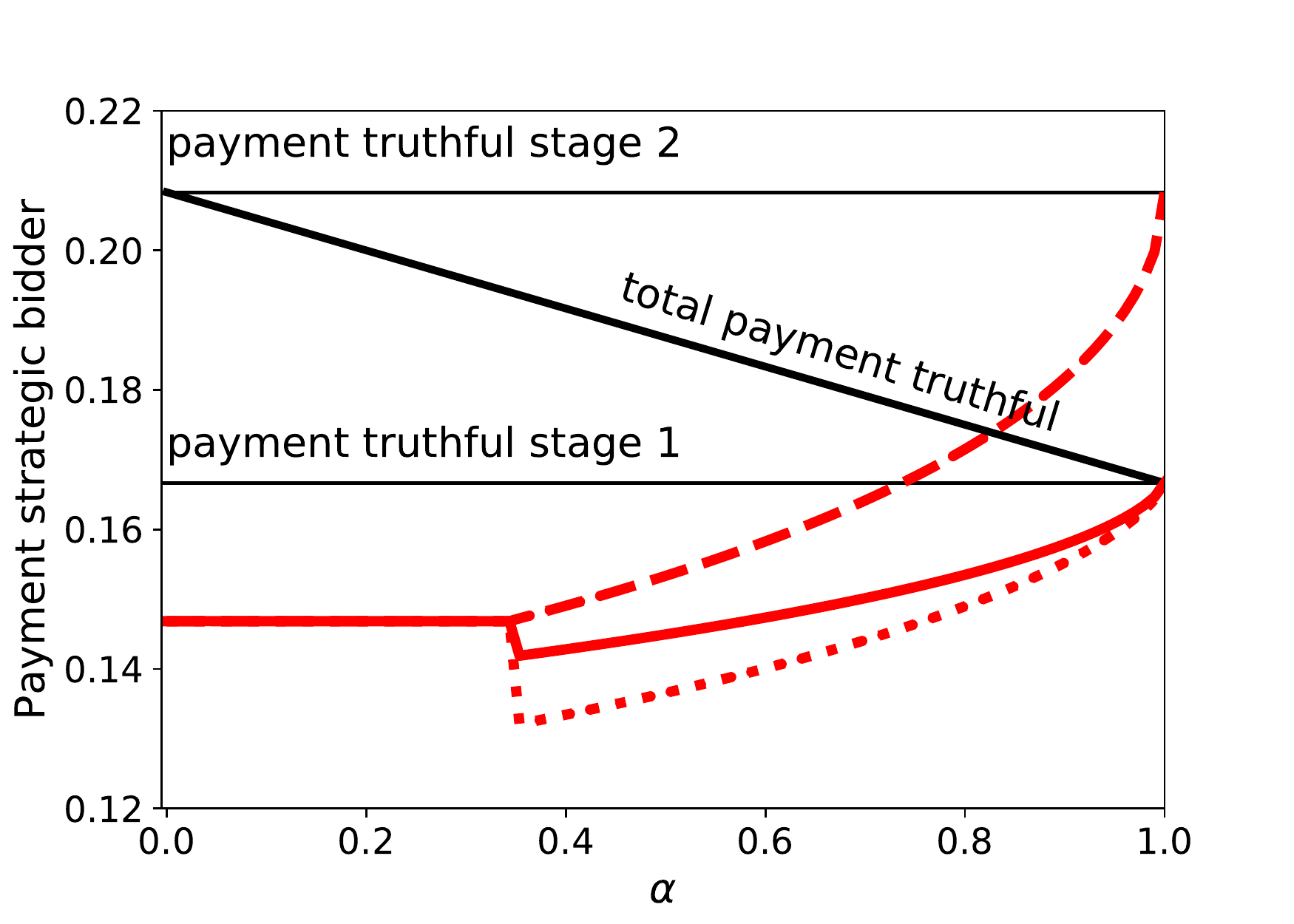}
  \includegraphics[width=0.32\textwidth]{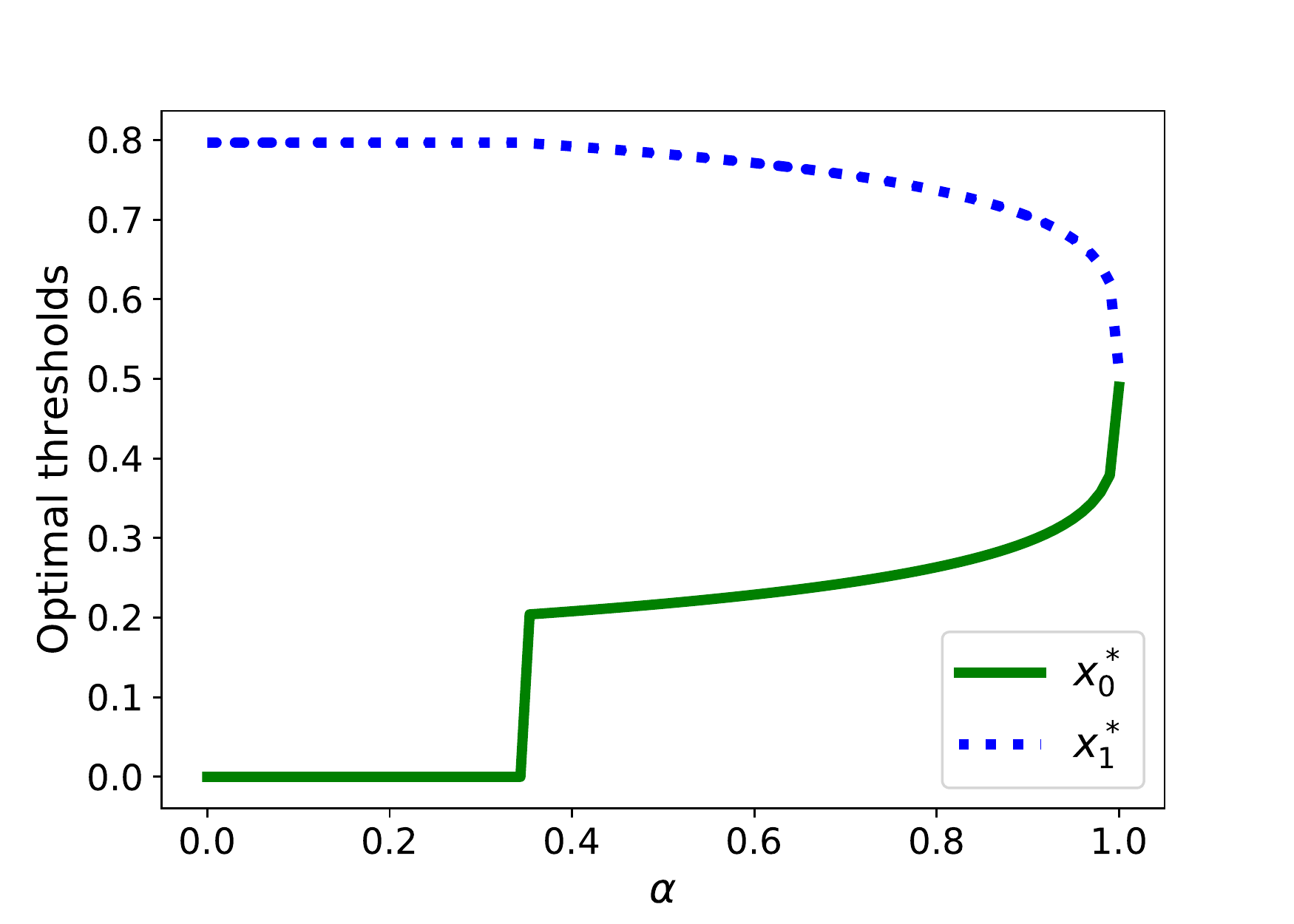}
  \caption{Welfare sharing with random reserve in stage 1: Left -- Utility of strategic bidder in the two-stage game as a function of $\alpha$. Middle -- Payment of strategic bidder in the two-stage game as a function of $\alpha$. The dotted lines corresponds to phase 1, the dashed ones to phase 2 and the plain line for the full game. The blue color is used for the bidder's utility and the red for her payment. The black color is used for the baseline of truthful bidding. Right: evolution of $x_0^*, x_1^*$ with $\alpha$.}
  \label{fig:welfare_vs_alpha}
\end{figure}
The results are presented in Fig. \ref{fig:welfare_vs_alpha}.
First, we recover the results from the literature corresponding to the extremes values of $\alpha$. For $\alpha = 1^-$, the payment and utility in the second stage are the ones of a 2$^\text{nd}$-price with monopoly price of value distribution. As $\alpha$ decreases, the total payment over the two stages quickly decreases (and utility increases) towards values that are even more favorable to the bidder than the ones of a 2$^\text{nd}$-price with no reserve price. This is consistent with the observations from \cite{nedelec2018thresholding} (case $\alpha=0$): if only one bidder is strategic, the payment of this strategic bidder can decrease even further than the one of a 2$^\text{nd}$-price with no reserve price. 
In Appendix \ref{ssec:plot_with_random_reserve}, we provide the same study with random reserve and show that the presence of exploration of the reserve prices in the first period (difference between left and right) does not impact much the payment/utility of the second period. It slightly decreases the advantage of the bidder, but not significantly.

Finally, we observe two regimes for $x_0^*$ (Fig. \ref{fig:welfare_vs_alpha}, right): for $\alpha < 0.8$ it is 0 and for $\alpha > 0.8$ it is the value that saturates the inequality constraint described by Th. \ref{dfn:thresholded_beta} between $x_0$ and $x_1$. This illustrates why we cannot only use first-order optimality condition: the optimal value is often reached on the edge of the feasible domain. We formalize this latter observation and shows that when $\alpha$ is small enough, the reserve value $x_0^*$ is exactly 0.
\begin{lemma}
Given a two-step process $(G, H, r_{\beta}^*, F)$ with $r^*_\beta = \inf \argmax_r R_B(r)$. If $\exists y \in [0,b], \forall x<y, \frac{G(x)H(x)}{f(x)} < x$ and $\frac{G(x)H(x)}{f(x)}$ is bounded on $[0,x^*]$ with $x^*=\inf\argmax_x R_X(x)$, then $\exists \alpha_0 > 0$ such that $\forall \alpha < \alpha_0$, 
\begin{align*}
(x^*_0,x^*_1) \in \argmax_{x_0, x_1}U_\alpha(\dtbeta) 
~~~ \Rightarrow ~~~ 
x_0^* = 0
\end{align*}
\label{thm:x0_to_0}
\end{lemma}
\vspace{-1em}
The proof can be found in Appendix \ref{ssec:welfare_sharing_best_response}.  $\frac{G(x)H(x)}{f(x)}$ being bounded on $[0,x^*]$ is satisfied if $F$ is MHR (monotonous hazard rate). $\frac{G(x)H(x)}{f(x)}<x$ is the key assumption. Intuitively, it can be interpreted as \emph{"if the bidder has enough mass under the competition, she can overbid at almost no cost on small values to decrease the reserve value to 0."}
This theorem shows that the \emph{thresholding} strategies proposed in \cite{tang2016manipulate,nedelec2018thresholding} are actually optimal for $\alpha$ small enough and not just for $\alpha = 0$. In the following subsections, we carry out a formal analysis for such strategies $\tilde{\beta}_{0, x_1} = \tilde{\beta}_{x_1}$ as we can more easily characterize the best response and prove Nash equilibria.

\subsection{Phase transition for the classical thresholding strategies in the two-stage game} 
\label{subsec:PhaseTransitionThresholding}
We suppose in this subsection that the seller commits to thresholding, i.e. $x_0=0$ with the notation above. The seller uses monopoly reserves in the second phase. 
\begin{theorem}\label{thm:TwoStepProcessCommitment}
We call ${\mathcal G}_\alpha=\alpha G_1+(1-\alpha) G_2$, $G_1$ and $G_2$ being the distributions of the competition faced by the bidder in the two phases, possibly including reserve prices. We assume $\psi^{-1}(0)>0$.
If the buyer uses the thresholding strategy $\tbeta$ of Theorem \ref{dfn:thresholded_beta}  and commits to it, we have for their utility, if the seller is welfare benevolent~:
the utility has (in general) a discontinuity at $x_1=\psi^{-1}(0)$. For $x_1<\psi^{-1}(0)$, we have $U(0)\geq U(x_1)$. For $x_1>\psi^{-1}(0)$, the first order condition are the same as in \cite{nedelec2018thresholding,tang2016manipulate} where the distribution of the competition is now ${\mathcal G}_\alpha=\alpha G_1+(1-\alpha) G_2$. Call $x_1^*(\alpha)$ the unique solution of this problem. Hence, the optimal threshold is $\argmax_{x_1\in \{0,x_1^*(\alpha)\}} U(x_1)$. An optimal threshold at $x_1=0$ corresponds to bidding truthfully. 

\end{theorem}

\begin{lemma}[Phase transition in $\alpha$]\label{lemma:phaseTransition}
Assume the setup of Theorem \ref{thm:TwoStepProcessCommitment}. Suppose that $G_1=G_2$, i.e. the distribution of the competition is the same in the two phases and there is no reserve price in the first phase.  Then there is a critical value $\alpha_c$ for which, if $\alpha<\alpha_c$, it is preferable for the bidder to threshold and if $\alpha>\alpha_c$, it is preferable for the bidder to bid truthfully. 	
\end{lemma}
The proofs of the theorem and the lemma can be found in Section \ref{sec:ProofsSec3Thresholding} and we provide more details on the utility there. 
A graphical illustration of these results can be found in the Appendix, Figure \ref{fig:2StageNoReservesCommitment}.

\subsection{Existence of a Nash equilibrium in the two-stage game} 
\label{subsec:ExistenceNashEquilibriumIn2Stage}
Most of the literature on strategic buyers in repeated auctions have focused on the posted price setting \cite{amin2013learning,medina2014learning}, though see \cite{golrezaei2018dynamic} for a  recent extension to second price auction. To complement this line of work, we exhibit now the existence of a Nash equilibrium between strategic bidders in the two-stage game. Our approach here is the following. All players have essentially two strategies, according to Theorem \ref{thm:TwoStepProcessCommitment}: they can either threshold optimally above the monopoly price or bid truthfully. We consider the first order conditions for a Nash equilibrium if all players threshold and we then study the best response of the remaining player. 

\begin{theorem}\label{thm:ExistenceNashEqui2Stage}
We consider the symmetric case where all bidders have the same value distribution $F$ and it is regular. We assume for simplicity that the distribution is supported on [0,1]. Suppose this distribution has density that is continuous at 0 and 1 with $f(1)\neq 0$ and  $\psi$ crosses 0 exactly once and hence is positive beyond that crossing point. 

In the case where $\alpha=0$, there is a unique symmetric Nash equilibrium. The revenue of the bidders in this case is the same as in a second price auction with no reserve. 

For any fixed $\alpha>0$, there exists $\alpha_{c,\text{thresh}}$ (possibly 0 or $\infty$) such that if $\alpha\leq \alpha_{c,\text{tresh}}$ there exists a symmetric Nash equilibrium in the class of thresholded strategies. It is the same as in the case $\alpha=0$ and hence the revenue of the bidders  is again the same as in a second price auction with no reserve. 

There exists $\alpha_{c,\text{truthful}}$ (possibly 0 or $\infty$) such that if $\alpha\geq \alpha_{c,\text{truthful}}$ there exists a symmetric Nash equilibrium in the class of thresholded strategies and it corresponds to bidding truthfully. 	
\end{theorem}
The proof for the case $\alpha=0$ is in Section \ref{sec:proofNashAlphaZero}. The remainder of the proof is in Section \ref{sec:ProofsSec3Thresholding}. The presence of other strategic bidders does not prevent bidders from taking advantage of the learning stage of the seller.  

\section{Robustness of bidding strategies to the information structure of the game}
We study the robustness of the bidding strategies to several variants of the two-stage game: 1) the seller is using an approximation of the bid distribution to compute the reserve price of the second stage, 2) the bidder does not have an estimation of the competition to compute the optimal strategy for the two-stage game and 3) the seller replaces the lazy second price with monopoly reserve with another type of auction. Their simplicity and robustness make them relevant in real-world interactions.
\subsection{Robustness to sample approximation of the seller and ERM/optimization algorithms}

\label{sec:empirical_estimation}
In practice, the seller needs to estimate the distribution of the buyer and hence does not have a perfect knowledge of the bid distribution $F_{B_i}$. The buyer needs to find a robust shading method, making sure that the seller has an incentive to lower her reserve price, even if she misestimates the bid distribution.

We focus on the specific problem of empirical risk minimization (ERM)\cite{morgenstern2015pseudo,dhangwatnotai2015revenue,huang2018making}. In Appendix \ref{app:subsec:proofERMRobustness}, we give some results about the impact of empirically estimating the value distribution $F$ by the empirical cumulative distribution function $\hat{F}_n$ on setting the reserve price. However because our approximations are formulated in terms of hazard rate, applying those results would yield quite poor approximation results in the context of setting the monopoly price through ERM. This is due to the fact that estimating a density pointwise in supremum norm is a somewhat difficult problem in general, associated with poor rates of convergence \cite{TsybakovBook2009}.  We now show how to take advantage of the characteristics of ERM to obtain better results than would have been possible by only considering the approximation error made on the virtual value. 


\begin{theorem}\label{thm:epsAndMonopolyPriceByERM}
Suppose the buyer has a continuous and increasing value distribution $F$, supported on $[0,b]$, $b\leq \infty$, with the property that if $r\geq y\geq x$, $F(y)-F(x)\geq \gamma_F (y-x)$, where $\gamma_F>0$. Suppose that $\sup_{t\geq r}t(1-F(t))=r(1-F(r))$. Suppose the buyer uses the strategy $\newThreshNotation$ defined as
$$
\newThreshNotation(x)=\left(\frac{(r-\epsilon)(1-F(r))}{1-F(x)}+\epsilon\right)\indicator{x\leq r}+x\indicator{x>r}\;,
$$ Assume she samples $n$ values $\{x_i\}_{i=1}^n$ i.i.d according to the distribution $F$ and bids accordingly in second price auctions. Call $x_{(n)}=\max_{1\leq i \leq n}x_i$. In this case the (population) reserve value $x^*$ is equal to 0. Assume that the seller uses empirical risk minimization to determine the monopoly price in a (lazy) second price auction, using these $n$ samples. Call $\hat{x}_n^*$ the reserve value determined by the seller using ERM. We have, if $C_n(\delta)=n^{-1/2}\sqrt{\log(2/\delta)/2}$ and $\epsilon>x_{(n)} C_n(\delta)/F(r)$ with probability at least $1-\delta_1$, 
$$
\hat{x}_n^*<\frac{2rC_n(\delta)}{\epsilon\gamma_F} \text{ with probability at least } 1-(\delta+\delta_1) \;.
$$

In particular, if $\epsilon$ is replaced by a sequence $\epsilon_n$ such that $n^{1/2}\epsilon_n min(1,1/x_{(n)}) \rightarrow \infty$ in probability, 
$\hat{x}_n^*$ goes to 0 in probability like $n^{-1/2}\max(1,x_{(n)})/\epsilon_n$. 
\end{theorem}
Informally speaking, our theorem says that using the strategy $\newThreshNotationEpsn$  with $\epsilon_n$ slightly larger than $n^{-1/2}$ will yield a reserve value arbitrarily close to 0. Hence the population results we derived in earlier sections apply to the sample version of the problem. We give examples and discuss our assumptions in Appendix \ref{app:subsec:proofERMRobustness} where we prove the theorem.
%

\subsection{Robustness to the knowledge of the competition distribution}
\label{sec:robust_competition}
Another common situation is that bidders do not know in advance the competition $G$ they are facing. They need to estimate it from past interactions with the seller and other bidders. First, we show that there exists some strategies that do not need a precise estimation of the competition. Then, we look at some worst-case scenario where the goal is to find the strategy with the highest utility in the worst case of G. We are optimizing the bidding strategy in the class of thresholded strategies. 

\begin{theorem}[Thresholding at the monopoly price]\label{thm:settingReserveValToZeroImprovesPerformance}
Consider the one-strategic setting in a lazy second price auction with $F_{X_i}$ the value distribution of the strategic bidder $i$ with a seller computing the reserve prices to maximize her revenue. Consider $\beta_{tr}$ the truthful strategy. Consider the case of $\alpha = 1$.
Then there exists $\beta$  which does not depend on G and   : 
\textbf{1)} $U_i(\beta)\geq U_i(\beta_{tr})$, $U_i$ being the utility of the strategic bidder.
\textbf{2)} $R_i(\beta)\geq R_i(\beta_{tr})$, $R_i$ being the payment of bidder $i$ to the seller. 
Then,  $\tilde{\beta}_{\psi^{-1}(0)}^{(\epsilon)}$ fulfill these conditions for $\epsilon\geq 0$ small enough.

For $\epsilon=0$, we call this strategy the thresholded strategy at the monopoly price.
In the two-stage game, there is a critical value $\alpha_c$ for which, if $\alpha<\alpha_c$, it is preferable for the bidder to threshold at the monopoly price and if $\alpha>\alpha_c$, it is preferable for the bidder to bid truthfully. 	
\end{theorem}
The formal proof in the general case is in Appendix \ref{proof_Section6}. The nice and crucial property of thresholding at the monopoly price is that it does not depend on a specific knowledge of the competition. 
\paragraph{Numerical applications in the case $\alpha=0$:} In the case of two bidders with uniform value distribution, the strategic bidder utility increases from 0.083 to 0.132, (a 57\% increase). In the case of two bidders with exponential distribution,with parameters $\mu =  0.25$ and $\sigma = 1$, the utility of the strategic bidders goes from  0.791  to 1.025 (a  29.5\% increase). 
This theorem shows that even if a fine-tuned knowledge of the competition helps  bidders increase their utility, there exists a strategy where bidders do not need to know the competition and can still significantly increase their utility. Due to lack of space, we provide more results on the impact of knowledge of the competition in Appendix \ref{proof_Section6}.
\subsection{Robustness to a change of mechanism}
\label{sec:change_mechanism}
We finally show that the thresholded strategies are robust to certain changes of mechanism. This makes sense in the framework of Stackelberg games with no commitment where the player performs a certain optimization given the objective function of the other player whereas this other player can decide to change her optimization problem. In this section, we study the impact on the utility of the second stage if the seller decides to change from a lazy second price with monopoly reserve to another type of revenue-maximizing auction. 

\textbf{Myerson auction.} 
In the one-strategic setting, in the symmetric case where all bidders have initially the same value distribution, we show that the utility gain for the strategic bidder of thresholding at the monopoly price is higher in the Myerson auction than in the lazy second price auction.
\begin{lemma}
\label{lemma_Myerson}
Consider the case where the distribution of the competition is fixed. Assume all bidders have the same value distribution $F_X$, and that $F_X$ is regular. Assume that bidder $i$ is strategic and that the $K-1$ other bidders bid truthfully. Let us denote by $\beta_{truth}$ the truthful strategy and $\beta_{thresh}$ the thresholded strategy at the monopoly price. The utility of bidder $i$ in the Myerson auction $U_i^{Myerson}$ and in the lazy second price auction $U_i^{Lazy}$ satisfy
\begin{equation*} 
U_i^{\text{Myerson}}(\beta_{\text{thresh}}) - U_i^{\text{Myerson}}(\beta_{\text{truth}}) \geq U_i^{\text{Lazy}}(\beta_{\text{thresh}}) - U_i^{\text{Lazy}}(\beta_{\text{truth}}).
\end{equation*}
\end{lemma}
The proof is given in Appendix \ref{appendix_myerson}. Numerics with $K=2$ bidders with $\mathcal{U}[0,1]$ value distribution: The utility is 1/12 in the truthful case in both lazy second price and Myerson since these auctions are identical in the symmetric case. The utility of the thresholded strategy at monopoly price is $7/48$ in the Myerson auction, i.e. $75\%$ more than the utility with the truthful strategy. We note that the gain is larger than for a second price auction with monopoly reserve where the extra utility was $57\%$.

\textbf{Eager second price auction with monopoly price.} 
Monopoly reserves are not optimal reserve prices for this version of the second price auction in general but in practice, the optimal ones are NP-hard to compute \cite{wang2016optimal,paes2016field}. We recall that the eager second price auction is a standard second price auction between bidders who clear their reserves.

\begin{lemma}
Consider the same setting and notations as Lemma \ref{lemma_Myerson}. The utility of bidder $i$ in the eager second price auction with monopoly reserves, $U_i^{Eager}$, and in the lazy second price auction $U_i^{Lazy}$ with the same reserves, satisfy : $U_i^{\text{Eager}}(\beta_{\text{thresh}}) - U_i^{\text{Eager}}(\beta_{\text{truth}}) \geq U_i^{\text{Lazy}}(\beta_{\text{thresh}}) - U_i^{\text{Lazy}}(\beta_{\text{truth}})$.
\end{lemma}
The proof is given in Appendix \ref{appendix_eager}.
These two lemmas show that thresholded strategies can increase the utility of strategic bidders even if the seller runs a different auction in the different stage than the lazy second price auction with monopoly price. In future work, we plan to design games between bidders and seller where the seller can change the mechanism at any point and bidders can update their bidding strategy changing the bid distribution observed by the seller. 
\section{Conclusion}
Our paper contributes to the growing literature on learning in the presence of strategic behavior. It extends recent work on repeated auctions \cite{amin2014repeated,mohri2015revenue} and addresses a key issue recognized in these papers, namely understanding the interaction between sellers and multiple strategic buyers. We propose novel optimal strategies for the canonical problem of auction design with unknown value distributions that are learned in an exploration phase and exploited thereafter.  
The Stackelberg game we consider exhibit complex solutions but we provide  simple strategies that can be made robust to various setups reflecting how much information bidders have about the game they are participating in. This allows buyers to quantify the price of revealing information about their values in repeated auctions. It also opens new avenues for research on the seller side by providing new realistic strategies that may adopted by the strategic bidder. 

\section*{Acknowledgements}
V. Perchet has benefited from the support of the FMJH Program Gaspard Monge in optimization and operations research (supported in part by EDF) and from the CNRS through the PEPS program. N. El Karoui gratefully acknowledges support from grant NSF DMS 1510172.

\newpage
\bibliographystyle{unsrt}
\bibliography{main}
\newpage
\appendix 
\section{Proof of Myerson Lemma}
\label{appendix_lemma1}
\begin{customlemma}{1}[Integrated version of the Myerson lemma]
Let bidder $i$ have value distribution $F_i$ and call $\beta_i$ her strategy, $F_{B_i}$ the induced distribution of bids and $\psi_{B_i}$  the corresponding virtual value function. Assume that  $F_{B_i}$ has a density and finite mean. Suppose that $i$'s bids are independent of the bids of the other bidders and denote by $G_i$ the cdf of the maximum of their bids. Suppose a lazy second price auction with reserve price denoted by $r$ is run. Then the payment $M_i$ of bidder $i$ to the seller can be expressed as 
\begin{equation*}
M(\beta_i) = \mathbb{E}_{B_i \sim F_{B_i}}\bigg(\psi_{B_i}(B_i)G_i(B_i)\textbf{1}(B_i \geq r)\bigg)\;.
\end{equation*}
When the other bidders are bidding truthfully, $G_i$ is the distribution of the maximum value of the other bidders.
\end{customlemma}
\begin{proof}
The proof is similar to the original one \cite{Myerson81} (see \cite{krishna2009auction} for more details).
It consists in using Fubini's theorem and integration by parts (this is why we need conditions on $F_{B_i}$)  to transform the standard form of the seller revenue, i.e.
$$
\mathbb{E}_{B_i\sim F_{B_i},B_j\sim F_{B_j}}\bigg(\max_{j\neq i}(B_j,r)\indicator{B_i\geq \max_{j\neq i}(B_j,r)}\bigg)
$$
into the above equation. 
We consider a lazy second price auction. Call $r$ the reserve price for bidder $i$. So the seller revenue from bidder $i$ is 
$$
\Exp{\max_{j\neq i}(B_j,r)\indicator{B_i\geq \max_{j\neq i}(B_j,r)}}
$$
In general, we could just call $Y_i=\max_{j\neq i}B_j$ and say that the revenue from $i$, or $i$'th expected payment is 
$$
\Exp{\max(Y_i,r)\indicator{B_i\geq \max(Y_i,r)}}\;.
$$
Call $G$ the cdf of $Y_i$ and $\tilde{G}$ the cdf of $\max(Y_i,r)$. Note that $\tilde{G}(t)=\indicator{t \geq r}G(t)$.

So if we note $B_i=t$, we have
$$
\Exp{\max(Y_i,r)\indicator{B_i\geq \max(Y_i,r)}|B_i=t}=\int_0^t u d\tilde{G}(u)
$$
Integrating by parts we get 
\begin{align*}
\int_0^t u d\tilde{G}(u)&=\left. u\tilde{G}(u)\right|_0^t-\int_0^t \tilde{G}(u)du\\
&=t\tilde{G}(t)-\int_0^t \tilde{G}(u)du=\indicator{t\geq r}\left[t G(t)-\int_{r}^t G(u)du\right]
\end{align*}
Hence, 
$$
\Exp{\max(Y_i,r)\indicator{b_i\geq \max(Y_i,r)}}=\int_0^\infty 
\left[\indicator{b_i\geq r}b_i G(b_i)-\int_{0}^{b_i}\indicator{u\geq r} G(u)du\right] f_{B_i}(b_i)db_i
$$
The first term of this integral is simply 
$$
\int_0^\infty \indicator{b_i\geq r}b_i G(b_i)f_{B_i}(b_i)db_i=
\Exp{B_iG(B_i)\indicator{B_i\geq r}}.
$$
Note that to split the two terms of the integral we need to assume that $\Exp{B_i}<\infty$, hence the first moment assumption on $F_i$.
The other part of the integral is 
\begin{align}
&\int_0^\infty \left(\int_{0}^{b_i}\indicator{u\geq r} G(u)du\right)f_{B_i}(b_i)db_i=\int_0^\infty \left(\int \indicator{b_i\geq u}\indicator{u\geq r} G(u)du\right)f_{B_i}(b_i)db_i\\
&=\int\int\indicator{u\geq r} G(u)  \indicator{b_i\geq u} f_{B_i}(b_i)  db_i du
=\int\indicator{u\geq r} G(u)  \left(\int\indicator{b_i\geq u} f(b_i)  db_i\right) du\\
&=\int\indicator{u\geq r} G(u)  P(B_i\geq u)du
= \int\indicator{u\geq r} G(u)  \frac{1-F_{B_i}(u)}{f_{B_i}(u)}f_{B_i}(u)du\\
&=\Exp{\indicator{B_i\geq r}G(B_i)\frac{1-F_{B_i}(B_i)}{f_{B_i}(B_i)}}
\end{align}
We used Fubini's theorem to change order of integrations, since all functions are non-negative. The result follows.

Of course, when $f_{B_i}(b_i)=0$ somewhere we understand $f_{B_i}(b_i)/f_{B_i}(b_i)=0/0$ as being equal to 1. To avoid this problem completely we can also simply write 
$$
M(\beta_i)=\int \left[b_if_{B_i}(b_i)-(1-F_{B_i}(b_i))\right]G(b_i)\indicator{b_i\geq r}db_i=\int \frac{\partial [b_i(F_{B_i}(b_i)-1)]}{\partial b_i}G(b_i) \indicator{b_i\geq x} db_i\;.
$$
If $F_{B_i}$ is not differentiable but absolutely continuous, its Radon-Nikodym derivative is used when interpreting the differentiation of $[b_i(F_{B_i}(b_i)-1)]$ with respect to $b_i$. 

\subsection{Technical lemmas}
\begin{lemma}\label{definition_psi}
Suppose $B_i=\beta_i(X_i)$, where $\beta_i$ is increasing and differentiable and $X_i$ is a random variable with cdf $F_i$ and pdf $f_i$. Then 
\begin{equation}\label{eq:ODEPhiG}
h_{\beta_i}(x_i)\triangleq \beta_i(x)-\beta_i'(x)\frac{1-F_{i}(x)}{f_{i}(x)} = \psi_{F_{B_i}}(\beta_i(x)) \;.
\end{equation}
\end{lemma}
\begin{proof}
$\psi_{F_{B_i}}(b) = b - \frac{1-F_{B_i}(b)}{f_{B_i}(b)}$ with $F_{B_i}(b) = F_{i}(\beta_i^{-1}(b))$ and $f_{B_i}(b) = f_{i}(\beta_i^{-1}(b)/\beta_i'(\beta_i^{-1}(b)$.
\\
Then, 
$h_{\beta_i}(x) = \psi_{B_i}(\beta_i(x)) = \beta_i(x) - \beta_i'(x)\frac{1-F_{i}(x)}{f_{i}(x)}\;.$
\end{proof}

The second lemma shows that for any function $g$ we can find a function $\beta$ such that $h_\beta = g$.

\begin{lemma}\label{lemma:keyODEs}
Let $X$ be a random variable with cdf $F$ and pdf $f$. Assume that $f>0$ on the support of $X$. Let $x_0$ in the support of $X$, $C\in\mathbb{R}$ and $g:\mathbb{R} \rightarrow \mathbb{R}$. lf we note 
\begin{equation}\label{eq:gAsConditionalExpectation}
\beta_g(x)=\frac{C(1-F(x_0))-\int_{x_0}^{x} g(u) f(u) du}{1-F(x)}\;.
\end{equation}


Then, if $B=\beta_g(X)$, 
$$
h_{\beta_g}(x) = g(x) \text{ and } \beta_g(x_0)=C\;.
$$ 
If $x_0\leq t$ and $g$ is non-decreasing on $[x_0,t]$,  $\beta_g'(x)\geq (C-g(x))(1-F(x_0))f(x)/(1-F(x))$ for $x\in[x_0,t]$. Hence $\beta_g$ is increasing on $[x_0,t]$ if $g$ is non-decreasing and $g<C$.
\end{lemma}
\begin{proof}
The result follows by simply differentiating the expression for $\beta_g$, and plugging-in the expression for $h_{\beta_g}$ obtained in Lemma \ref{definition_psi}. The result on the derivative is simple algebra.
\end{proof}

\end{proof}

\section{Optimal Classes of Strategies for two-step Process}

\paragraph{Disclaimer:} Our lines of proof here are very similar to the one in \cite{tang2016manipulate}, adapted to the two-stage process. The main difference is that in the two-stage process, the optimal reserve value is not always 0, which gives slightly different forms of optimal strategies for the regular case and make us introduce a quasi-regular case. Also, note that some of the proofs do not require assumptions as strong as we do in the paper on the continuity of F.

\paragraph{General two-stage process with commitment:} A general two-stage process with commitment is a tuple of the form $P= (G, H,r, F)$ that is defined as follow: 
\begin{description}
\item[1 -- exploration] The seller runs a lazy 2$^\text{nd}$-price auction without reserve price. The current bidder faces a competition $G$. In this step, the potential randomized or deterministic reserve price is denoted by the distribution $H$.
\item[2 -- exploitation] The seller runs a lazy 2$^\text{nd}$-price auction with reserve price $r$. The current bidder faces the same competition $G$ as in the first step.
\end{description}
The bidder has a value distribution $F$. In this process, the bidder commits to a strategy $\beta$ corresponding to a bid distribution $F_B$ that stays the same in both steps. We additionally assume that $F_B$ is also not too fat-tailed -- meaning $1-F_B(b) = o(b^{-1})$ to ensure the seller's revenue supremum is reached for at least some finite values of reserve price.

The two-step process described in the main body of the paper corresponds to a such a tuple $P_\beta = (G_1, G_2, r^*_B, F)$ where $r$ is chosen to maximize $R_B(r)$, the revenue of the seller defined as $R_B(r) = r \lP_{B\sim F_B}(r \leq B)$. We also assume the seller to be welfare-benevolent, meaning that $r^*_\beta = \min \argmax_r R_B(r)$.

\paragraph{Technical note on $R_B$.} If $F_B$ is not continuous, we need to make a difference between $F_B(b) = \lP_{B\sim F_B}(B \leq b)$ and $\tilde{F}_B(b) = \lP_{B\sim F_B}(B < b)$. Technically, $F_B$ is c\`adl\`ag as a repartition function, while $\tilde{F}_B$ is l\`adc\`ag. More practically, if the distribution contains a point mass, the value at which there is the point mass won't be on the same side of the discontinuity in $F_B$ and in $\tilde{F}_B$. In terms of definition of the seller's revenue $R_B$, as beating the reserve price $r$ is usually implemented as $\indicator{B \geq r}$, then when doing the integration, we can indeed find the revenue is defined as $R_B(b) = b (1-\tilde{F}_B(b))$.

First, we state and (re-)prove a well-known property that is central for our proofs: that when the reserve price does not depends on one bidder distribution, the closer her strategy to truthful, the better in utility. For the sake of simplicity, we consider in this statement the reserve price is part of the fixed competition.
\begin{lemma}
  Given a two-stage process $P = (G, H, r, F)$, for any $\beta, \rho$ and any $\alpha \in [0,1]$,
  $$\forall x \in [0, +\infty), \rho(x)(x - \beta(x)) \leq 0 ~~~\Rightarrow~~~ U_\alpha(\beta+\rho) \leq U_\alpha(\beta)$$
  \label{lem:variation_from_truthful}
\end{lemma}

\begin{proof}
Denoting $G_r$ is the cdf of $Y_r = \max(Y, r)$ and $G_\alpha(x) = \alpha G(x)H(x) + (1-\alpha)G_r(x)$, we have
\begin{align*}
    U_\alpha(\beta+\rho) - U_\alpha(\beta) 
    & = 
    \alpha\lE_{\substack{X\sim F\\Y\sim GH}}\left[(X-Y)\left(\indicator{\beta(X)+\rho(X) \geq Y} - \indicator{\beta(X) \geq Y}\right)\right]\\
    &\,\,\,\,\,+
    (1-\alpha) \lE_{\substack{X\sim F\\Y_r\sim G_r}}\left[(X-Y_r)\left(\indicator{\beta(X)+\rho(X) \geq Y_r} - \indicator{\beta(X) \geq Y_r}\right)\right]\\
    & = 
    \lE_{\substack{X\sim F\\Y\sim G_\alpha}}\left[(X-Y)\left(\indicator{\beta(X)+\rho(X) \geq Y > \beta(X)} - \indicator{\beta(X) \geq Y > \beta(X)+\rho(X)}\right)\right]\\
\end{align*}
\begin{itemize}
    \item if $\forall x, \beta(x)+\rho(x) \geq \beta(x) \geq x$ then $\forall x, (X-Y)\indicator{\beta(X)+\rho(X) \geq Y > \beta(X)} \leq 0$ (and the other indicator is 0)
    \item if $\forall x, \beta(x)+\rho(x) \leq \beta(x) \leq x$ then $\forall x, (X-Y)\indicator{\beta(X) \geq Y > \beta(X)+\rho(X)} \geq 0$ (and the other indicator is 0)
\end{itemize}
Both ways, $\forall x, \rho(x)(x - \beta(x)) \leq 0$, then $U_\alpha(\beta+\rho) - U_\alpha(\beta) \leq 0$.
\end{proof}

\subsection{Seller Revenue is Lower than with Truthful Bidding}
If the reader is intuitively convinced that there exists a best response strategy $\beta$ with $B=\beta(X)$ such that the monopoly revenue is not higher than the one of the truthful strategy -- i.e. $\sup_{r'} R_B(r') \leq \sup_{r'} R_X(r')$ -- you can safely skip this section. Only for this section, we are going to explicit the dependency of the utility on the two-stage process $P$ by using the notation $U_\alpha^P(\beta)$.

\begin{lemma}
Given a two-step process $P_\beta = (G, H, r^*_\beta, F)$ with $r^*_\beta = \inf \argmax_r R_B(r)$, 
\begin{align*}
\max_r R_B(r) > \max_r R_X(r) \,\,\Rightarrow\,\, \forall \alpha\in[0,1], \,U^{P_\beta}_\alpha(\beta) < \max_{\omega} U^{P_\omega}_\alpha(\omega)
\end{align*}
\label{lem:seller_revenue_decreases}
\end{lemma}

\begin{proof}
Assuming $\beta$ is such that $\max_r R_B(r) > R^* = \max_r R_X(r)$, we can build $\omega$ such that for any $x$ such that $\beta(x)(1-\tilde{F}(x)) \leq R^*, \,\omega(x) = \beta(x)$ and $\omega(x) = \frac{R^*}{1-\tilde{F}(x)}$ otherwise.

We can choose $\rho(x) = \beta(x) - \omega(x)$ for any $x$. For any $x$, either we have $x \geq \beta(x)$ and $\rho(x) \geq 0$ or $\rho(x) = 0$. Hence, for any $x$, $(x - \beta(x)) \rho(x) \geq 0$. By Lemma \ref{lem:variation_from_truthful}, we have that in the process $P_\omega$ (the one with reserve price $r_\omega^*$ in the second step), $U^{P_\omega}_\alpha(\beta) \leq U^{P_\omega}_\alpha(\omega)$.

Then, because $r^*_\omega \leq r^*_\beta$, we have $U^{P_\beta}_\alpha(\beta) \leq U^{P_\omega}_\alpha(\beta)$.
\end{proof}

\subsection{Deriving the Best Quasi-Regular Response}
\label{ssec:best_quasi_regular_response}
Because the revenue curve generated by the best response \cite{tang2016manipulate} has some limiting properties in practice, we propose to further constrain $\beta$ by asking $R_B$ to be quasi-concave when $R_X$ is so.

\begingroup
\def\thetheorem{\ref{dfn:thresholded_beta}}
\begin{theorem}
Given a 2-step process $(G, H, r_{\beta}^*, F)$ with $r^*_\beta = \inf \argmax_r R_B(r)$ and such that $F$ is quasi-regular\footnote{We say a distribution is quasi-regular if the associated revenue curve $R_{X_i}$ is quasi-concave}, $\forall \alpha\in[0,1], \,\exists 0 \leq x_0 \leq x_1$ such that the best quasi-regular response (maximizing $U_\alpha(\beta)$ with $F_B$ regular) is \begin{align}\label{eq:defDoubleThresh}
\dtbeta(x)
=
\indicator{x \leq x_0}x
+
\indicator{x_0 < x \leq x_1}\frac{R}{1-F(x)}
+
\indicator{x > x_1}x
~~~~~\text{where}~~~~~ R = x_1 (1 - F(x_1))
\end{align}
Moreover, we have $x_1 = \sup \{x: x(1-F(x)) \geq R\}$ and $x_0 \leq \overline{x}_0 = \inf \{x: x(1-F(x)) \geq R\}$.
\end{theorem}
\addtocounter{theorem}{-1}
\endgroup
\begin{proof}
Thanks to Lemma \ref{lem:seller_revenue_decreases} and to the Assumption that we have $1-F_B(b) = o(b^{-1})$, we can restrict ourselves w.l.o.g. to $\beta$ with a finite reserve price $r^*_\beta$. As $R_B$ is l\`adc\`ag and upper semicontinuous (because $\tilde{F}_B$ is so), then $\exists x_r, r^*_\beta = \beta(x_r)$.

We denote $R = \sup_x (x \land \beta(x))(1-\tilde{F}(x))$ an define 
\begin{align*}
    \omega(x) = 
    \begin{cases}
    \frac{R}{1-\tilde{F}(x)} \text{ if } \beta(x)(1-\tilde{F}(x)) \leq R \leq x(1-\tilde{F}(x))\\
    \frac{R}{1-\tilde{F}(x)} \text{ if } \beta(x)(1-\tilde{F}(x)) \geq R \geq x(1-\tilde{F}(x))\\
    x \text{ otherwise}
    \end{cases}
\end{align*}
By definition of $\omega$, denoting $\rho(x) = \beta(x) - \omega(x)$, we have that for any $x$, $\rho(x)(x-\omega(x)) \leq 0$. Thus, by Lemma \ref{lem:variation_from_truthful}, $U^{(P_\beta)}_\alpha(\omega) \leq U^{(P_\beta)}_\alpha(\beta)$. 

Moreover, $\omega(x_r)(1-\tilde{F}(x_r)) = R$, thus $r^*_\omega \leq r^*_\beta$ and $U^{(P_\omega)}_\alpha(\omega) \leq U^{(P_\beta)}_\alpha(\omega)$. In the end, we have $U^{(P_\omega)}_\alpha(\omega) \leq U^{(P_\beta)}_\alpha(\beta)$.

We can use Lemma \ref{lem:convexity_of_union_of_shit} to show that $\{x: \omega(x) = \frac{R}{1-\tilde{F}(x)}\}$ is convex and the decrease rate of the right  tail to say it is bounded, hence $\exists x_1 \geq x_0 \geq 0, ~~\{x: \omega(x) = \frac{R}{1-\tilde{F}(x)}\} = [x_0, x_1]$.

Finally, we need to check the conditions on $x_0$ and $x_1$.
By definition of $R$ and $x_0$, we have $x_0 \leq \inf \{x: x(1-\tilde{F}(x)) \geq R\}$.
Then, if $x_1 > \sup \{x: x(1-\tilde{F}(x)) \geq R\}$, we can define $\tilde{x}_1 = \sup \{x: x(1-\tilde{F}(x)) \geq R\}$, then define $\tilde{\omega}$ equal to $\omega$ below $\tilde{x}_1$ and equal to the identity above. Then we can use again Lemma \ref{lem:variation_from_truthful} to show the utility is improved as  the reserve price didn't change.
\end{proof}

\begin{lemma}
Given two real-values functions $f$ and $g$ quasi-concave. we denote $R = \sup_x (f\land g)(x)$. We have
$$
\{x: f(x) \geq R\}\cup\{x: g(x) \geq R\} \text{ is convex.}
$$
\label{lem:convexity_of_union_of_shit}
\end{lemma}
\begin{proof}
First, both $\{x: f(x) \geq R\}$ and $\{x: g(x) \geq R\}$ are convex ($f,g$ are quasi-concave). By contradiction, we assume there exists $x < z < y$ such that $f(x) \geq R$, $g(y)\geq R$, $f(z) < R$ and $g(z) < R$. Because $g$ is increasing before $z$, $\sup_{x < z}(f\land g)(x) \leq g(z) < R$. With the same argument over $f$, we have $\sup_{x > z}(f\land g)(x) \leq f(z) < R$.
\end{proof}

\subsection{Welfare Sharing between Seller and Bidder}
\label{ssec:welfare_sharing_best_response}
\begingroup
\def\thetheorem{\ref{thm:double_thresholded_derivatives}}
\begin{theorem}
Given a two-step process $(G, H, r_{\beta}^*, F)$ with $r^*_\beta = \inf \argmax_r R_B(r)$, the best response is of the form of $\tilde{\beta}_{x_0^*, x_1^*}$ and the utility $U_\alpha(\dtbeta)$ has the following derivatives:
\begin{align}
    \frac{\partial U_\alpha(\dtbeta\!)}{\partial x_0} \!=\! \alpha(1-F(x_0))G(x_0)H(x_0)
    -
    x_0 f(x_0) G_\alpha\!\!\left(\frac{x_1(1-F(x_1))}{(1-F(x_0))}\right)~~~~~~~~~~~\\
    \frac{\partial U_\alpha(\dtbeta\!)}{\partial x_1} 
    \!= \!
    f(x_1)\psi(x_1)
    \!\left(\!G_\alpha(x_1) - \lE_X\!\!\left(\!\indicator{x_0 \leq X \leq x_1}\frac{X}{(1-F(X))}g_\alpha\!\!\left(\!\frac{x_1(1-F(x_1))}{(1-F(X))}\!\right)\!\!\right)\!\!\right)
\end{align}
where $G_\alpha(x) = \alpha G(x)H(x) + (1-\alpha)G(x)$.
\end{theorem}
\addtocounter{theorem}{-1}
\endgroup
\begin{proof}
The push-forward of $F$ through $\dtbeta$ has the following virtual value:
$$
\psi_B(\beta(X)) = \indicator{X \leq x_0}\psi(X) + \indicator{x_0 < X \leq x_1}0 + \indicator{X > x_1}\psi(X)
$$
Then, we can use Myerson's Lemma \cite{Myerson81} to write the utility:
\begin{align}
    U_\alpha(\dtbeta) 
    = &
    ~\lE_X\left(\indicator{X<x_0}(X - \psi(X))\alpha G(x)H(x)\right)\\
    &+
    \lE_X\left(\indicator{x_0 < X \leq x_1}X G_\alpha\left(\frac{x_1(1-\tilde{F}(x_1))}{1-\tilde{F}(x)}\right)\right)\nonumber\\
    &+
    \lE_X\left(\indicator{X<x_1}(X - \psi(X))G_\alpha(X)\right)\nonumber
\end{align}
Straigtforward calculations gives the derivatives.
\end{proof}

\begingroup
\def\thetheorem{\ref{thm:x0_to_0}}
\begin{theorem}
Given a two-step process $(G, H, r_{\beta}^*, F)$ with $r^*_\beta = \inf \argmax_r R_B(r)$. We assume here $\forall x\in (0,x^*], f(x) > 0$ with $x^*=\argmax_x R_X(x)$. If $\exists y \in (0,x^*], \forall x<y, \frac{G(x)H(x)}{f(x)} < x$ and $\frac{G(x)H(x)}{f(x)}$ is bounded on $(0,x^*]$, then $\exists \alpha_0 > 0$ such that $\forall \alpha < \alpha_0$, 
$$
(x^*_0,x^*_1) \in \argmax_{x_0, x_1}U_\alpha(\dtbeta) 
~~~ \Rightarrow ~~~ 
x_0^* = 0
$$
\end{theorem}
\addtocounter{theorem}{-1}
\endgroup
\begin{proof}
Denoting $\overline{x}_1=\sup_\alpha x^*_1$,
\begin{align}
    \frac{\partial U_\alpha(\dtbeta)}{\partial x_0}
    &=
    \alpha(1-F(x_0))G(x_0)H(x_0)
    -
    x_0 f(x_0) G_\alpha\left(\frac{x_1(1-F(x_1))}{(1-F(x_0))}\right)\\
    &\leq 
    \alpha G(x_0)H(x_0)
    -
    x_0 f(x_0) G_\alpha\left(\overline{x}_1(1-F(\overline{x}_1))\right)\\
    &\leq 
    \alpha G(x_0)H(x_0)
    -
    x_0 f(x_0) \inf_{\gamma \in [0,1]}G_\gamma\left(\overline{x}_1(1-F(\overline{x}_1))\right)\\
\end{align}
So $\frac{\partial U_\alpha(\dtbeta)}{\partial x_0} \leq 0 \Leftrightarrow \alpha \frac{G(x_0)H(x_0)}{f(x_0)} \leq x_0 \inf_{\gamma \in [0,1]}G_\gamma\left(\overline{x}_1(1-F(\overline{x}_1))\right)$.
Finally, as $\inf_{\gamma \in [0,1]}G_\gamma\left(\overline{x}_1(1-F(\overline{x}_1))\right) \in (0, b)$ (from Lemma \ref{lem:x_1_bounded}), I can choose $\alpha_0 = 1/\inf_{\gamma \in [0,1]}G_\gamma\left(\overline{x}_1(1-F(\overline{x}_1))\right)$.

\textbf{Case $x_0 < y$:} from the assumptions, we can directly get $\frac{\partial U(x_0, x_1)}{\partial x_0} \leq 0$.

\textbf{Case $y \leq x_0$:} Denoting $C = \sup_{x\in[y, x^*]} \frac{G(x)H(x)}{f(x)}$, we can denote $\alpha_1 = \frac{y}{x \alpha_0}$.

Finally, for $\alpha < \alpha_2 = \min(\alpha_0,\alpha_1)$, we have $\frac{\partial U_\alpha(\dtbeta)}{\partial x_0} \leq 0$ for any $x_0$.
\end{proof}
\begin{lemma}
Given a two-step process $(G, H, r_{\beta}^*, F)$ with $r^*_\beta = \inf \argmax_r R_B(r)$, we have $\overline{x}_1=\sup_\alpha x^*_1 < b$.
\label{lem:x_1_bounded}
\end{lemma}
\begin{proof}
For this we study the different terms of $U_\alpha$ when $x_1$ goes to $b$, to prove that in such case, the utility goes to 0.
\begin{align*}
\lE_{X}\left(\indicator{X\leq x_0^*}(X - \psi_{X}(X))\alpha G(X)H(X)\right) &\leq \int\indicator{X\leq x_0^*}(1-F(X))G(X)H(X)\alpha{\rm d}X\\
&\leq \int\indicator{X\leq x_0^*}{\rm d}X\\
& \leq x_0^* = o_{x_1^* \to b}(1)
\end{align*}
\begin{align*}
\lE_{X}\left(\indicator{X> x_1^*}(X - \psi_{X}(X)) G_\alpha(X)\right) &\leq \int\indicator{X > x_1^*}(1-F(X))G_\alpha(X)H(X){\rm d}X\\
&\leq \int\indicator{X > x_1^*}(1-F(X)){\rm d}X\\
& = o_{x_1^* \to b}(1) ~~~\text{as either $1-F(x) = o(x^{-1})$ or $b$ is finite.}
\end{align*}
Fix $y\in (x_0^*, x_1^*)$,
\begin{align*}
    \lE_X\left(\indicator{x_0^* < X \leq x_1^*}X G_\alpha\left(\frac{x_1^*(1-F(x_1^*))}{1-F(x)}\right)\right)
    & =  \lE_X\left(\indicator{x_0^* < X \leq y}X G_\alpha\left(\frac{x_1^*(1-F(x_1^*))}{1-F(x)}\right)\right)\\
    &~~~~+
    \lE_X\left(\indicator{y < X \leq x_1^*}X G_\alpha\left(\frac{x_1^*(1-F(x_1^*))}{1-F(x)}\right)\right)
\end{align*}
We have $x_1^*(1-F(x_1^*)) \to 0$ when $x_1^* \to b$ as either $b$ is finite or $1-F(x) = o(x^{-1})$. Hence, we have $\lE_X\left(\indicator{x_0^* < X \leq y}X G_\alpha\left(\frac{x_1^*(1-F(x_1^*))}{1-F(x)}\right)\right) = o_{x_1^* \to b}(1)$.
On the other side,
\begin{align*}
\lE_X\left(\indicator{y < X \leq x_1^*}X G_\alpha\left(\frac{x_1^*(1-F(x_1^*))}{1-F(x)}\right)\right) &\leq \int_y^b X f(X){\rm d}X \\
& = y(1-F(y)) + \int_y^b 1-F(X){\rm d}X \\
& = o_{y \to b}(1)
\end{align*}
By contradiction, assume $\forall y < b, \exists \alpha \in [0,1]$ such that $x_1^* > y$. Then for any $\varepsilon > 0$, we can set $y$ such that $\exists \alpha, U_\alpha(\tilde{\beta}_{x_0^*, x_1^*}) \leq \varepsilon$, which contradict the optimality of $(x_0^*, x_1^*)$ (as bidding truthfully brings a strictly positive utility).
\end{proof}
\subsection{Figures when the first stage is with random reserve}
The distribution of reserve price in stage 1 is $\mathcal{U}([0,1])$.
\label{ssec:plot_with_random_reserve}
\begin{figure}[h!]
  \centering
  \includegraphics[width=0.48\textwidth]{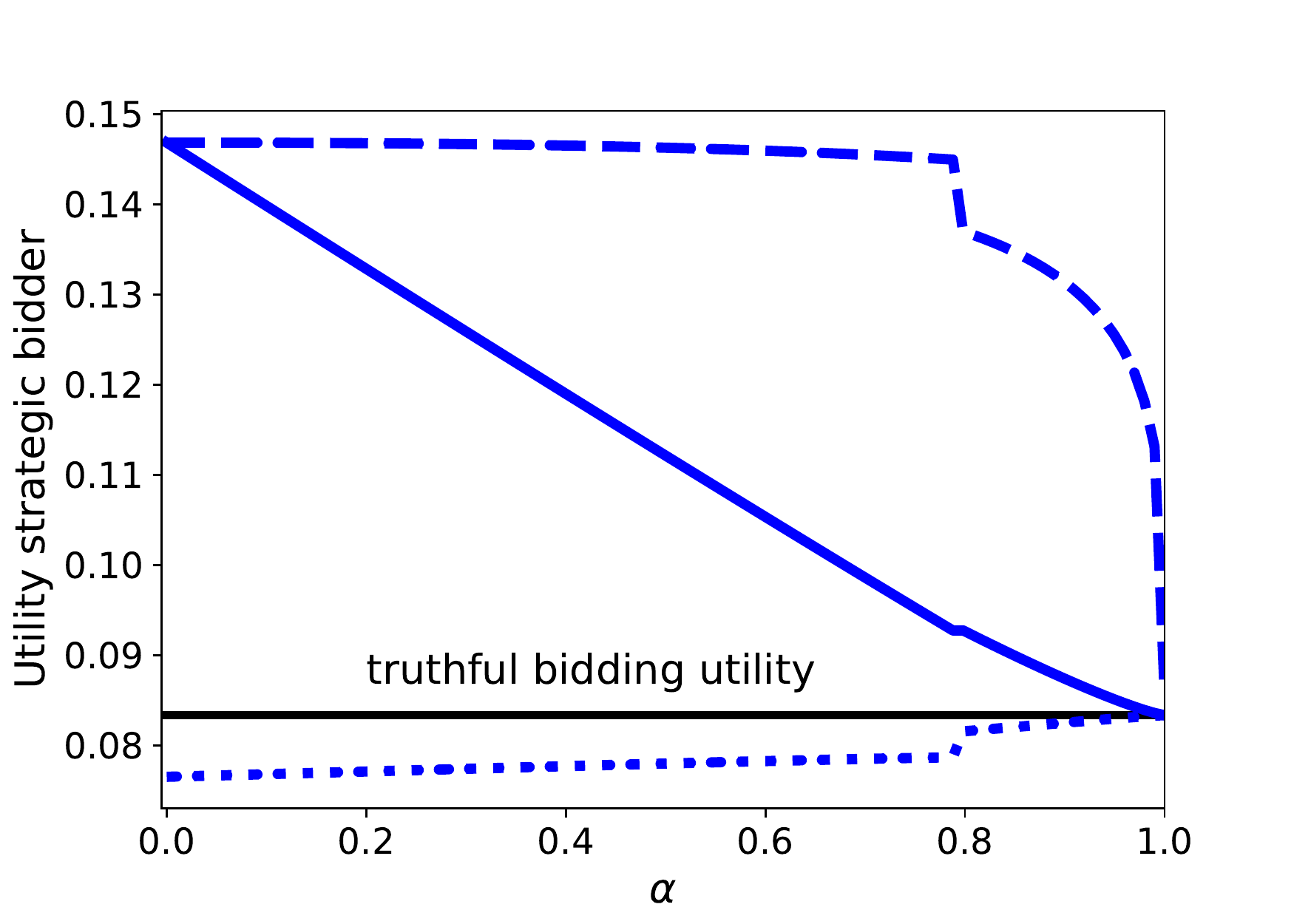}
    \includegraphics[width=0.48\textwidth]{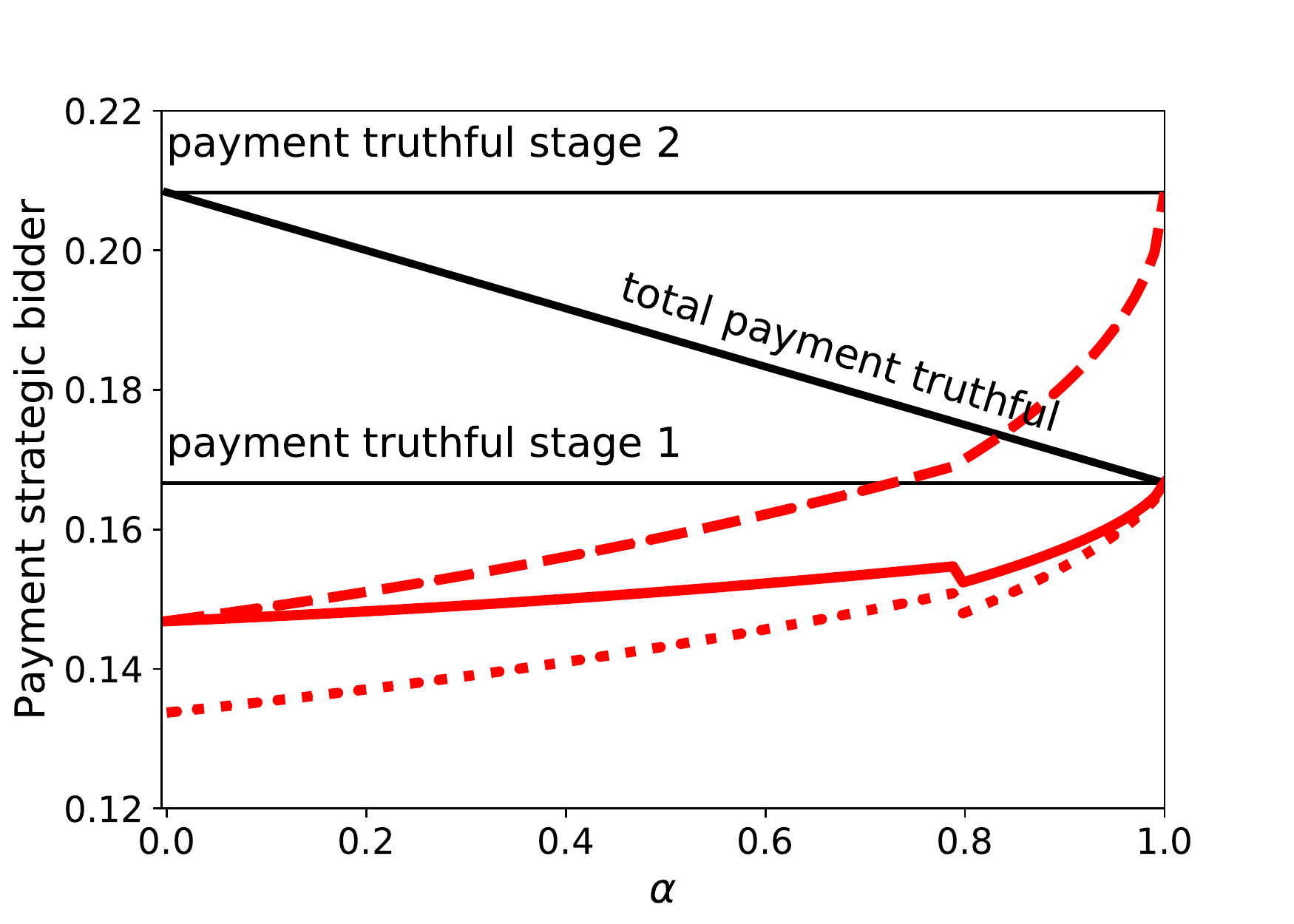}
  \includegraphics[width=0.48\textwidth]{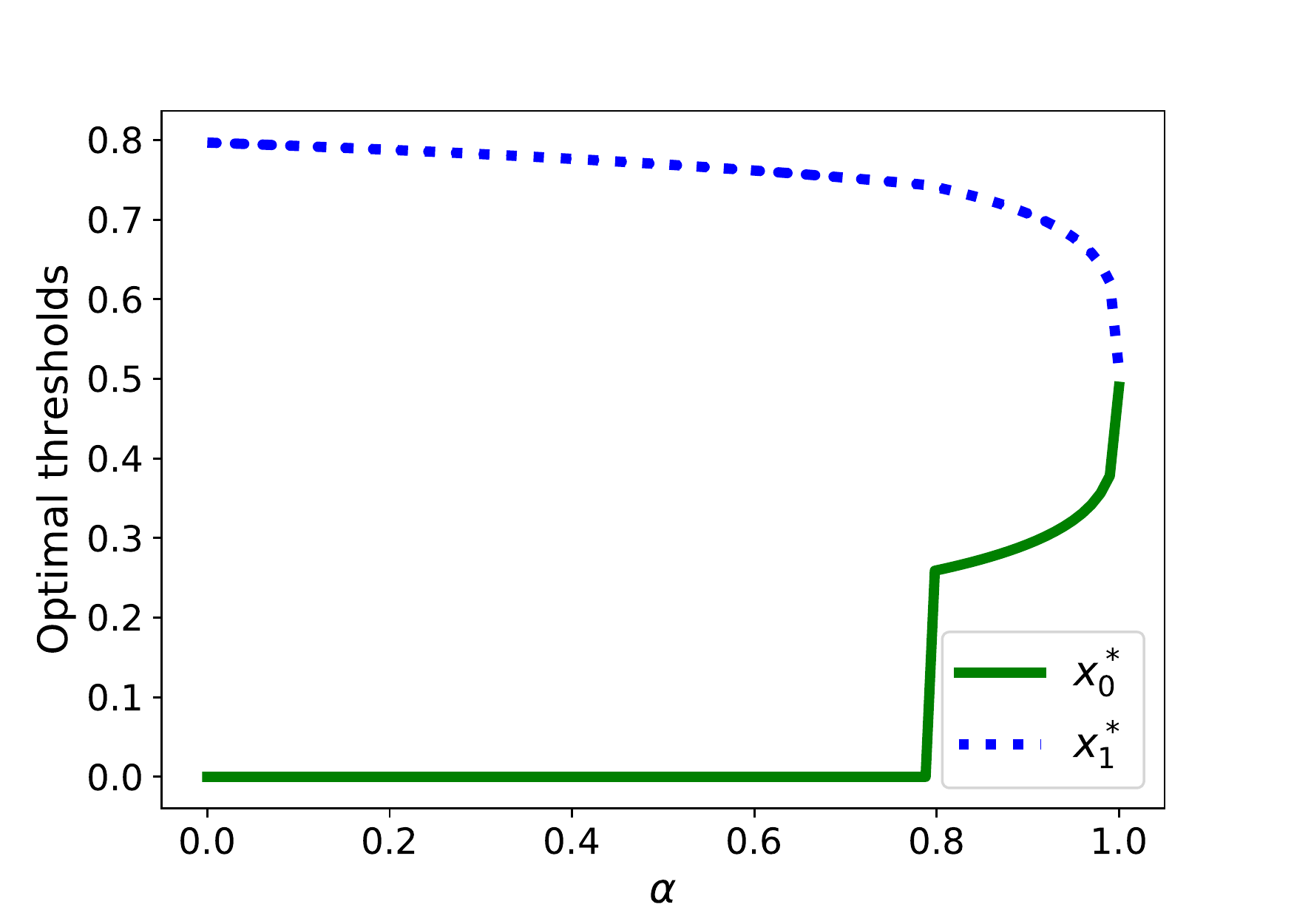}
  \caption{Welfare sharing with random reserve in stage 1: Left -- Utility of the strategic bidder function of $\alpha$. Middle -- Payment of the strategic bidder function of $\alpha$.  The dotted lines corresponds to phase 1, the dashed ones to phase 2. The blue color is used for the bidder's utility and the red for her payment. The black color is used for the baseline of truthful bidding. Right: evolution of $x_0^*, x_1^*$ with $\alpha$.}
  \label{fig:welfare_vs_alpha}
\end{figure}
\newpage

\section{Proof of the existence of a Nash equilibrium ($\alpha=0$)} 
\label{sec:proofNashAlphaZero}
More formally, the game we are considering is the following: all bidders have the same value distribution $F_X$. The mechanism is a lazy second price auction with reserve price equal to the monopoly price of bidders' bid distributions. We denote by $\beta_1,...,\beta_n$ the bidding strategies used by each player and by $U_i(\beta_1,...,\beta_n)$ the utility  
of bidder $i$ when using $\beta_i$. We say that strategy $\beta$ is a symmetric Nash equilibrium for this game if for all bidding strategies $\hat{\beta}$ and for all bidders $i$, if all players except bidder $i$ are using $\beta$, 
\begin{equation*}
U_i(\beta,...,\hat{\beta}, ...,\beta) \leq U_i(\beta,...,\beta, ...,\beta)
\end{equation*}  

We can reuse the same type of formulas we introduced in the previous subsection to compute the expected utility of each bidder. We focus on bidder $i$. We note $G_\beta$ the distribution of the maximum bid of the competitors of bidder $i$. The distribution of the highest bid of the competition depends now on the strategy of bidder i.  With this notation,
\begin{equation*}
U_i(\beta,...,\hat{\beta},...,\beta) = \mathbb{E}\bigg((X-h_{\hat{\beta}}(X))G_\beta(\hat{\beta}(X))\textbf{1}(X \geq h_{\hat{\beta}}^{-1}(0))\bigg)\;.
\end{equation*}
As in the previous setting, we can compute directional derivatives when $\hat{\beta} = \beta + t\rho$. We can find different sorts of Bayes-Nash equilibria depending on the class of function bidders are optimizing in. We call them \emph{restricted Nash equilibrium} since they are equilibrium in a restricted class of bidding strategies. In practice, bidders may not be able or willing to implement all possible bidding strategies. They often limit themselves to strategies that are easy to practically implement.

We have obtained results on \emph{restricted Nash equilibria} when the bidders restrict themselves to the class of linear shading or affine shading, though we do not present them here as they are a bit tangential to the main thrust of this paper. We detail the results on a larger and more interesting class of functions: the class of thresholded bidding strategies that we introduced in Definition \ref{dfn:thresholded_beta}. We assume that bidders can choose any strategies in this large class of functions. 

Our first theorem on a Nash equilibrium between bidders is the following: 
\begin{theorem}
\label{thm:nashEquiThresholded}
We consider the symmetric setting where all the bidders have the same value distributions. We assume for simplicity that the distribution is supported on [0,1]. Suppose this distribution has density that is continuous at 0 and 1 with $f(1)\neq 0$ and X has a distribution for which $\psi_X$ crosses 0 exactly once and is positive beyond that crossing point. 

There is a unique symmetric Nash equilibrium in the class of thresholded bidding strategies defined in Definition \ref{dfn:thresholded_beta}.
It is found by solving Equation \eqref{eq:keyEqThreshStratAllStrategic} to determine $r^*$ and bidding truthfully beyond $r^*$

Moreover, if all the bidders are playing this strategy corresponding to the unique Nash equilibrium, the revenue of the seller is the same as in a second price auction with no reserve price. The same is true of the utility of the buyers. 	
	
\end{theorem}
Sketch of proof: 
\begin{itemize} 
\item State a directional derivative result in this class of strategies by directionally-differentiating the expression of $U_i(\beta,...,\beta_t,...,\beta)$.
\item Prove uniqueness of the Nash equilibrium.
\item Prove existence of the Nash equilibrium.
\item Show equivalence of revenue with a second price auction without reserve.
\end{itemize}

\subsection{A directional derivative result}
We can state a directional derivative result in this class of strategies by directionally-differentiating the expression of $U_i(\beta,...,\beta_t,...,\beta)$.It implies the only strategy with 0 ``gradient" in this class is truthful beyond a value $r$ ($r$ is different from the one in strategic case), where $r$ can be determined through a non-linear equation. 
\begin{lemma}\label{lemma:directionalDerivativesThresholdedStrategiesSymmetric}
Consider the symmetric setting. Suppose that all the bidders are strategic and that they are all except one using the strategy $\beta$, $G_\beta$ denotes the CDF of the maximum bid of the competition when they use $\beta$ and $g_\beta$ the corresponding pdf, consider r and $\gamma$ such that $\beta = \beta_r^\gamma$. Assuming that the seller is welfare-benevolent and that $\beta_i^t=\beta_r^{\gamma+t\rho}$. Then  
\begin{align*}
&\left.\frac{\partial U(\beta,...,\beta_i^t,...,\beta)}{\partial t}\right|_{t=0}=\mathbb{E}\bigg((X-\gamma(X))g_{\beta}(\gamma(X))\rho(X)\indicator{X\geq r}\bigg) +\rho(r)\\
&(1-F(r))\bigg(\mathbb{E}\bigg(\frac{X}{1-F(X)}g_\beta\left(\frac{\gamma(r)(1-F(r))}{1-F(X)}\right)\indicator{X\leq r}\bigg)-G_\beta(\gamma(r))\bigg)
\end{align*}
Also, 
\begin{align*}
&\frac{\partial U(\beta,...,\beta_i,...,\beta)}{\partial r}=-h_\beta(r)f(r)\bigg(\mathbb{E}\bigg(\frac{X}{1-F(X)}g_\beta\left(\frac{\gamma(r)(1-F(r))}{1-F(X)}\right)\indicator{X\leq r}\bigg)-G_\beta(\gamma(r))\bigg)
\end{align*}
The only strategy $\gamma$ and threshold $r$ for which we can cancel the derivatives in all directions $\rho$ consists in bidding truthfully beyond $r^*_{all}$, where $r^*_{all}$ satisfies the equation 
\begin{equation}\label{eq:keyEqThreshStratAllStrategic}
\frac{K-1}{r^*_{all}(1-F(r^*_{all}))}\mathbb{E}\bigg(XF^{K-2}(X)(1-F(X))\indicator{X\leq r^*_{all}}\bigg) = F^{K-1}(r^*_{all})\;.
\end{equation}
\end{lemma}
\begin{proof}
In the case of $K$ symmetric bidders with the same value distributions $F$, we have 
\begin{equation*}
G_\beta(x)= F^{K-1}(\beta^{-1}(x))\;\qquad\text{and}\qquad g_{\beta}(\beta(x))=\frac{(K-1)F^{K-2}(x)}{\beta'(x)}f(x)\;.
\end{equation*}
The last result follows by plugging-in these expressions in the corresponding equations.
\end{proof}

\subsection{Uniqueness of the Nash equilibrium}
We show that this strategy represents the unique symmetric Bayes-Nash equilibrium in the class of shading functions defined in Definition \ref{dfn:thresholded_beta}. At this equilibrium, the bidders recover the utility they would get in a second price auction without reserve price.

\begin{lemma}\label{lemma:keyEqThreshAllStratHasUniqueSolution}
Suppose $X$ has a distribution for which $\psi_X$ crosses $0$ exactly once and is positive beyond that crossing point. Then Equation \eqref{eq:keyEqThreshStratAllStrategic} has a unique non-zero solution. 
\end{lemma}
\begin{proof}
We have by integration by parts  
\begin{align*}
&(K-1)\Exp{\bigg(XF^{K-2}(X)(1-F(X))\indicator{X\leq r}\bigg)}
\\&=r(1-F(r))F^{K-1}(r)+\Exp{\bigg(\vValue_X(X) F^{K-1}(X)\indicator{X\leq r}\bigg)}\;.
\end{align*}
Hence finding the root of 
$$
\frac{K-1}{r(1-F(r))}\Exp{\bigg(XF^{K-2}(X)(1-F(X))\indicator{X\leq r}\bigg)} = F^{K-1}(r)
$$
amounts to finding the root(s) of 
$$
\mathfrak{R}(r)\triangleq\Exp{\bigg(\vValue_X(X) F^{K-1}(X)\indicator{X\leq r}\bigg)}=0\;.
$$
0 is an obvious root of the above equation but does not work for the penultimate one.
Note that for the class of distributions we consider (which is much larger than regular distributions but contains it), this function $\mathfrak{R}$ is decreasing and then increasing after $\vValue_X^{-1}(0)$, since the virtual value is negative and then positive. Since $\mathfrak{R}(0)=0$, it will have at most one non-zero root for the distributions we consider. The fact that this function is positive at infinity (or at the end of the support of $X$) comes from the fact that its value then is the revenue-per-buyer of a seller performing a second price auction with $K$ symmetric buyers bidding truthfully with a reserve price of 0. And this is by definition positive. 
So we have shown that for regular distributions and the much broader class of distributions we consider the function $\mathfrak{R}$ has exactly one non-zero root. 
\end{proof}
\subsection{Existence of the Nash equilibrium}
\subsubsection{Best response strategy: one strategic case} 
\begin{definition}[Thresholded bidding strategies]\label{definition_thresholded_strategies}
 A bidding strategy $\beta$ is called a thresholded bidding strategy if and only if there exists $r>0$ such that for all $x < r, h_{\beta}(x) = \psi_B(\beta(x)) = 0$. This family of functions can be parametrized with
\begin{equation*}
\beta_r^\gamma(x)=\frac{\gamma(r)(1-F(r))}{1-F(x)}\indicator{x< r}+\gamma(x)\indicator{x\geq r}\;,
\end{equation*}
with $r \in \mathbb{R}$ and $\gamma: \mathbb{R} \rightarrow \mathbb{R}$ strictly increasing.
\end{definition}
\begin{lemma}
\label{lemma:directionalDerivativesThresholdedStrategiesRight}
Suppose that only one bidder is strategic, let G denote the CDF of the maximum value of the competition and g the corresponding pdf. Denote by $U(\beta_r^\gamma)$ the utility of the bidder using the strategy $\beta_r^\gamma$ according to the parametrization of Definition \ref{definition_thresholded_strategies}. Assume that the seller is welfare-benevolent. Then if $\beta_t=\beta_r^{\gamma+t\rho}$,
\begin{align*}
&\left.\frac{\partial U(\beta_t)}{\partial t}\right|_{t=0}=\mathbb{E}\bigg((X-\gamma(X))g(\gamma(X))\rho(X)\indicator{X\geq r}\bigg) +\rho(r)\\
&(1-F(r))\bigg(\mathbb{E}\bigg(\frac{X}{1-F(X)}g\left(\frac{\gamma(r)(1-F(r))}{1-F(X)}\right)\indicator{X\leq r}\bigg)-G(\gamma(r))\bigg)
\end{align*}
\begin{align*}
\frac{\partial U(\beta_r^\gamma)}{\partial r}=&-h_\beta(r)f(r)
\\&\bigg(\mathbb{E}\bigg(\frac{X}{1-F(X)}g\left(\frac{\gamma(r)(1-F(r))}{1-F(X)}\right)\indicator{X\leq r}\bigg)-G(\gamma(r))\bigg)
\end{align*}
The only strategy $\gamma$ and threshold $r$ for which we can cancel the derivatives in all directions $\rho$ consists in bidding truthfully beyond $r^*$, where $r^*$ satisfies the equation 
\begin{equation}\label{eq:keyEqThreshStratOneStrategic}
G(r^*)=\mathbb{E}\bigg(\frac{X}{1-F(X)}g\left(\frac{r^*(1-F(r^*))}{1-F(X)}\right)\indicator{X\leq r^*}\bigg)\;.
\end{equation}
\end{lemma}
\begin{proof} 
Recall that our revenue in such a strategy (we take $\eps=0$) is just, when the seller is welfare-benevolent (and hence s/he will push the reserve value to 0 as long as $\vValue_B(\beta(x))\geq 0$ for all $x$)
$$
U(\beta_r^\gamma)=\Exp{(X-\vValue_B(\beta(X)))G(\beta(X))\indicator{X\geq r}}+\Exp{XG(\frac{\beta(r)(1-F(r))}{1-F(x)})\indicator{X\leq r}}\;.
$$
Now if $\beta_t=\beta+t \rho$, as usual, we have $\vValue_{B_t}(\beta_t)=h_\beta+t h_\rho$. Because we assumed that $\vValue_B(\beta(r))>0$, changing $\beta$ to $\beta_t$ won't drastically change that; in particular if $\vValue_{B_t}(\beta_t(r))>0$ the seller is still going to take all bids after $\beta_t(r)$.  In particular, we don't have to deal with the fact that the optimal reserve value for the seller may be completely different for $\beta$ and $\beta_t$. So the assumption $\vValue_B(\beta(r))>0$ is here for convenience and to avoid technical nuisances. In any case, we have 
\begin{gather*}
\frac{\partial U(\beta_t)}{\partial t}=\Exp{\left[-h_\rho(X)G(\beta(X))+(X-h_\beta(X))\rho(X)g(\gamma(X))\right]\indicator{X\geq r}}\\
+\Exp{X\rho(r)\frac{1-F(r)}{1-F(X)}g(\frac{\gamma(r)(1-F(r))}{1-F(x)})\indicator{X\leq r}}\;.	
\end{gather*}
Now using integration by parts on $\int_r^{\infty} \rho'(x)(1-F(x))G(\gamma(x))dx$, we have 
\begin{align*}
\Exp{-h_\rho(X)G(\beta(X))\indicator{X\geq r}}&=\left.\rho(x)(1-F(x))G(\gamma(x))\right|_{r}^\infty\\
&+\int_r^{\infty}[\rho(x)f(x)G(\gamma(x))-\rho(x)\beta'(x)g(\gamma(x))(1-F(x))]dx\\
&-\int_r^\infty \rho(x)f(x) G(\gamma(x))dx\\
&=-\rho(r)(1-F(r))G(\beta(r))-\Exp{\rho(X)\beta'(X)g(\gamma(X))\frac{1-F(X)}{f(X)}}\;.	
\end{align*}
Hence, 
\begin{gather*}
\frac{\partial U(\beta_t)}{\partial t}=\Exp{(X-\gamma(X))\rho(X)g(\gamma(X))\indicator{X\geq r}}-\rho(r)(1-F(r))G(\gamma(r))\\
+\rho(r)(1-F(r))\Exp{\left(\frac{X}{1-F(X)}\right)g(\frac{\gamma(r)(1-F(r))}{1-F(x)})\indicator{X\leq r}}\;.
\end{gather*}
This gives the first equation of Lemma \ref{lemma:directionalDerivativesThresholdedStrategiesRight}.

On the other hand, we have
\begin{gather*}
\frac{\partial U(\beta_r^\gamma)}{\partial r}=-(r-\vValue_B(\gamma(r)))G(\gamma(r))f(r)+rG(\gamma(r))f(r)\\
+\Exp{\frac{X}{1-F(X)}g(\frac{\gamma(r)(1-F(r))}{1-F(x)})\indicator{X\leq r}}\left[\gamma'(r)(1-F(r))-\beta(r)f(r)\right]\;.
\end{gather*}
Since 
$$
\left[\gamma'(r)(1-F(r))-\gamma(r)f(r)\right]=-f(r)\vValue_B(\beta(r))\;,
$$
we have established that 
\begin{gather*}
\frac{\partial U(\beta_r^\gamma)}{\partial r}=
\vValue_B(\gamma(r))f(r)\left[G(\gamma(r))-\Exp{\frac{X}{1-F(X)}g(\frac{\gamma(r)(1-F(r))}{1-F(x)})\indicator{X\leq r}}\right]\;.
\end{gather*}
We see that the only strategy $\beta$ and threshold $r$ for which we can cancel the derivatives in all directions $\rho$ consists in bidding truthfully beyond $r$, where $r$ satisfies the equation 
\begin{equation*}
G(r)=\Exp{\frac{X}{1-F(X)}g\left(\frac{r(1-F(r))}{1-F(X)}\right)\indicator{X\leq r}}\;.
\end{equation*}
Lemma \ref{lemma:directionalDerivativesThresholdedStrategiesRight} is shown.
\end{proof}
\subsubsection{Proof of existence of the Nash equilibrium}
$\bullet$ \textbf{On Equation \eqref{eq:keyEqThreshStratOneStrategic} and consequences}

Recall the statement of Equation \eqref{eq:keyEqThreshStratOneStrategic}. 
\begin{equation}\label{eq:keyeqOneStratExistence}
\Exp{\frac{X}{1-F(X)}g\left(\frac{r(1-F(r))}{1-F(X)}\right)\indicator{X\leq r}}=G(r)\;.
\tag{\ref{eq:keyEqThreshStratOneStrategic}}
\end{equation}
\begin{lemma}\label{lemma:atMostOneNonTrivialSoln}
Suppose $X$ has a regular distribution that is compactly supported (on [0,1] for convenience). 	
Equation  \eqref{eq:keyeqOneStratExistence} can be re-written as 
\begin{equation}\label{eq:keyeqOneStratExistenceAltForm}
\Exp{\frac{X}{1-F(X)}g\left(\frac{r(1-F(r))}{1-F(X)}\right)\indicator{X\leq r}}=G(r)+
\frac{1}{r(1-F(r))}\Exp{\vValue_X(X)G\left(\frac{r(1-F(r))}{1-F(X)}\right)\indicator{X\leq r}}\;.
\end{equation}
Equation \eqref{eq:keyeqOneStratExistence}	has at most one solution on (0,1). This possible root is greater than $\vValue^{-1}(0)$. 0 is also a (trivial) solution of  Equation \eqref{eq:keyeqOneStratExistence}.
\end{lemma}
The assumption that $X$ is supported on [0,1] can easily be replaced by the assumption that it is compactly supported, but it made notations more convenient. 
\begin{proof}[Proof of Lemma \ref{lemma:atMostOneNonTrivialSoln}]
We use again integration by parts~:
\begin{gather*}
\Exp{\frac{X}{1-F(X)}g\left(\frac{r(1-F(r))}{1-F(X)}\right)\indicator{X\leq r}}=
\int_0^r\frac{x(1-F(x))}{(1-F(x))^2}g\left(\frac{r(1-F(r))}{1-F(x)}\right)f(x) dx\\
=\left. x(1-F(x))\frac{G\left(\frac{r(1-F(r))}{1-F(x)}\right)}{r(1-F(r))}\right|_0^r
-\frac{1}{r(1-F(r))}\int_0^r (x(1-F(x)))' G\left(\frac{r(1-F(r))}{1-F(x)}\right) dx\\
=G(r)+\frac{1}{r(1-F(r))}\Exp{\vValue_X(X)G\left(\frac{r(1-F(r))}{1-F(X)}\right)\indicator{X\leq r}}\;.
\end{gather*}
So we are really looking at the properties of the solution of 
$$
\frac{1}{r(1-F(r))}\Exp{\vValue_X(X)G\left(\frac{r(1-F(r))}{1-F(X)}\right)\indicator{X\leq r}}=0\;.
$$
We call 
$$
I(r)=\Exp{\vValue_X(X)G\left(\frac{r(1-F(r))}{1-F(X)}\right)\indicator{X\leq r}}\;.
$$
For regular distributions, it is clear that if $r=\vValue^{-1}(0)$, $I(r)<0$.  

Now we note that 
$$
\frac{\partial I(r)}{\partial r}=\vValue_X(r)G\left(\frac{r(1-F(r))}{1-F(r)}\right)-\vValue_X(r)\Exp{\frac{\vValue_X(X)}{1-F(X)}g\left(\frac{r(1-F(r))}{1-F(X)}\right)\indicator{X\leq r}}\;.
$$
Using $\vValue_X(x)=x-(1-F(x))/f(x)< x$, we see that 
$$
\Exp{\frac{\vValue_X(X)}{1-F(X)}g\left(\frac{r(1-F(r))}{1-F(X)}\right)\indicator{X\leq r}}< \Exp{\frac{X}{1-F(X)}g\left(\frac{r(1-F(r))}{1-F(X)}\right)\indicator{X\leq r}}
$$
So if $r$ is a solution of 
$$
\Exp{\frac{X}{1-F(X)}g\left(\frac{r(1-F(r))}{1-F(X)}\right)\indicator{X\leq r}}=G(r)
$$
we have, if $r>\vValue_X^{-1}(0)$ and $r\neq 1$, 
$$
\frac{\partial I(r)}{\partial r}>0\;.
$$
So $r$ needs to be a solution of $I(r)=0$ (which is equivalent to the initial equation for non-trivial solutions) and must have $\frac{\partial I(r)}{\partial r}>0$.

So we have shown that $I$ is a function such that its (non-trivial) zeros are such that $I$ is strictly increasing at those roots. Because $I$ is differentiable and hence continuous, this implies that $I$ can have at most one non-trivial root. (0 is a trivial root of $I(r)=0$, though it is not an acceptable solution of our initial problem.)

We note that the end point of the support of $X$ is also a trivial solution of $I(r)=0$, by the dominated convergence theorem, though not an acceptable solution of our initial problem, as shown by a simple inspection.

\end{proof}

\begin{lemma}\label{lemma:uniqueSolnInOneStrategicCase}
Suppose that $G(0)=0$, $G$ is continuous and either $G(x)=g^{k}(0) x^k+o(x^k)$ near 0 for some $k$ or $G$ is constant near 0.  Assume $f$ has a continuous density near 1 with $f(1)\neq 0$ and the Assumptions of Lemma \ref{lemma:atMostOneNonTrivialSoln} are satisfied. Then Equation 
$$
\mathfrak{G}(r)=\frac{1}{r(1-F(r))}\Exp{\vValue_X(X)G\left(\frac{r(1-F(r))}{1-F(x)}\right)\indicator{X\leq r}}=0
$$
has a unique root in $(0,1)$. 

In particular in this situation there exists an optimal strategy in the class of shading functions defined in Definition \ref{dfn:thresholded_beta} and it is unique. It is defined by being truthful beyond $r^*: \mathfrak{G}(r^*)=0$ and shading in such a way that our virtual value is 0 below $r^*$. 	
\end{lemma}

\begin{proof}
We have already seen that this equation has at most one zero on (0,1) so we now just need to show that the function $\mathfrak{G}$ is positive somewhere to have established that it has a zero. Of course, for $r=\vValue^{-1}(0)$, the function is negative.

$\bullet$ \textbf{$\mathbf{G}$ not locally constant near 0}
Since $G$ is a cdf, and hence a non-decreasing function, the first $k$  such that $g^{(k)}(0)\neq 0$ has $g^{(k)}(0)>0$. Otherwise, $G$ would be decreasing around 0. We treat the case where $G$ is constant near 0 later so we now assume that $k$ exists and is finite.

Let $r$ be such that $1-F(r)=\eps$ very small (e.g. $10^{-6}$). Let $c<r$ be such that $(1-F(r))/(1-F(c))<\eta$ very small (e.g. $10^{-3}$) and $\vValue_X(c)>0$. Hence we can use a Taylor approximation to get that 
$$
G\left(\frac{r(1-F(r))}{1-F(x)}\right)\indicator{x\leq c}\simeq g^{(k)}(0)\frac{(r(1-F(r)))^k}{(1-F(x))^{k}} \indicator{x\leq c}\;.
$$
Integrating this out (ignoring at this point possible integration questions), we get 
$$
\Exp{\vValue_X(X)G\left(\frac{r(1-F(r))}{1-F(x)}\right)\indicator{X\leq c}}\simeq 
g^{(k)}(0)(r(1-F(r)))^k \Exp{\frac{\vValue_X(X)}{[1-F(X)]^k}\indicator{X\leq c}}\;.
$$
Now integration by parts shows that, if $k>1$ 
\begin{gather*}
\Exp{\frac{\vValue_X(X)}{[1-F(X)]^k}\indicator{X\leq c}}=\int_0^c x\frac{f(x)}{(1-F(x))^k}-\frac{1}{(1-F(x))^{k-1}}dx\\
=\left.\frac{xf(x)}{(1-F(x))^k}\right|_0^c-\left(1+\frac{1}{k-1}\right)\int_0^c\frac{dx}{[1-F(x)]^{k-1}}\;.
\end{gather*}
If $k=1$, using the fact that $(\ln(1-F(x))'=-f(x)/(1-F(x))$, we have 
\begin{gather*}
\Exp{\frac{\vValue_X(X)}{1-F(X)}\indicator{X\leq c}}=\int_0^c x\frac{f(x)}{1-F(x)}-1 dx\\
=\left.-x\ln(1-F(x))\right|_0^c- c+\int_0^c \ln(1-F(x))dx = -c\ln(1-F(c))-c+\int_0^c \ln(1-F(x))dx
\end{gather*}

We now assume that $k>1$; the adjustments for $k=1$ are trivial and are left to the reader. Clearly, when $F(c)$ is close to 1, we have, since we assume that $f(c)\neq 0$,  
$$
\int_0^c\frac{dx}{[1-F(x)]^{k-1}}\leq \frac{c}{(1-F(c))^{k-1}}=o\left(cf(c)/(1-F(x))^{k}\right)\;.
$$
So we have, as $c$ increases so that $F(c)\simeq 1$ (while of course having $(1-F(r))/(1-F(c))<\eta$),
$$
\frac{1}{r(1-F(r))}\Exp{\vValue_X(X)G\left(\frac{r(1-F(r))}{1-F(x)}\right)\indicator{X\leq c}}\sim 
[r(1-F(r))]^{k-1}\, g^{(k)}(0)\, \frac{cf(c)}{(1-F(c))^k}>0\;.
$$

So we have, if $\vValue_X(x)f(x)$ is continuous near $r$, i.e. $f(x)$ is continuous near $r$, 
\begin{gather*}
\Exp{\vValue_X(X)G\left(\frac{r(1-F(r))}{1-F(X)}\right)\indicator{c\leq X\leq r}}\geq G\left(\frac{r(1-F(r))}{1-F(c)}\right)\Exp{\vValue_X(X)\indicator{c\leq X\leq r}}\\
\simeq 
G\left(\frac{r(1-F(r))}{1-F(c)}\right)(r-c)(rf(r)-(1-F(r))
\simeq G\left(\frac{r(1-F(r))}{1-F(c)}\right)(r-c) rf(r)
\;.
\end{gather*}
Now we note that using a Taylor expansion of $1-F(x)$ around $r$, we have
$$
x\simeq r+ \frac{F(x)-F(r)}{f(r)}\;.
$$
So we see that 
\begin{gather*}
\frac{1}{r(1-F(r))}
\Exp{\vValue_X(X)G\left(\frac{r(1-F(r))}{1-F(X)}\right)\indicator{c\leq X\leq r}}\\
\simeq rG\left(\frac{r(1-F(r))}{1-F(c)}\right) \frac{F(r)-F(c)}{1-F(r)}
\simeq rG\left(\frac{r(1-F(r))}{1-F(c)}\right) \left(\frac{1-F(c)}{1-F(r)}-1\right)\;.
\end{gather*}
If now we take $c_2$ such that $1-F(r)/(1-F(c_2))=1/3$, the reasoning above applies and we have 
$$
\frac{1}{r(1-F(r))}
\Exp{\vValue_X(X)G\left(\frac{r(1-F(r))}{1-F(X)}\right)\indicator{c_2\leq X\leq r}}
\geq  rG\left(r/3\right) >0\;.
$$
Because $\pi(r)=(1-F(r))/(1-F(x))$ is increasing, we have $c\leq c_2$, since $\eta=\pi(c)\leq \pi(c_2)=1/3$ . So  we have 
$$
\frac{1}{r(1-F(r))}
\Exp{\vValue_X(X)G\left(\frac{r(1-F(r))}{1-F(X)}\right)\indicator{c\leq X\leq c_2}}\geq 0\;.
$$
Of course the choice of 1/3 above is arbitrary and it could be replaced by any fixed number $s<1$ such that $G(sr)>0$.
We conclude that $\mathfrak{G}$ is positive in a neighborhood of $1$. \\
$\bullet$ \textbf{$\mathbf{G}$ locally constant near 0}
In this case we can pick $c$ such that 
$$
\Exp{\vValue_X(X)G\left(\frac{r(1-F(r))}{1-F(x)}\right)\indicator{X\leq c}}=0\;.
$$
If $c$ is such that $\vValue(c)>0$ our arguments above immediately carry through. 
In fact we can ensure that this is always true by picking such a $c$ and picking a corresponding $r$ as function of the ratio $(1-F(r))/(1-F(c))$ we would like. 

So we conclude that even in this case, 
$\mathfrak{G}$ is positive in a neighborhood of $1$
\end{proof}

\begin{theorem}
We consider the symmetric case and assume that bidders have a compactly supported and regular distribution. We assume for simplicity that the distribution is supported on [0,1]. 

Suppose this distribution has density that is continuous at 0 and 1 with $f(1)\neq 0$. 

Then there is a unique symmetric equilibrium in the class of shading functions defined in Definition \ref{dfn:thresholded_beta}. 

It is found by solving Equation \eqref{eq:keyEqThreshStratAllStrategic} to determine $r^*$ and bidding truthfully beyond $r^*$.
	
\end{theorem}

\begin{proof}
We already know that there is at most one solution since Equation \eqref{eq:keyEqThreshStratAllStrategic} has exactly one solution. 	

If all the bidders but one put themselves at this strategy, we know from Lemma \ref{lemma:uniqueSolnInOneStrategicCase}, which applies because of our assumptions on $f$, that the optimal strategy for bidder one is unique in the class we consider and consists in using a shading that is truthful beyond $r$. This $r$ is uniquely determined by Equation \eqref{eq:keyEqThreshStratOneStrategic} but given the shading used by the other players we know that the $r$ determined by Equation \eqref{eq:keyEqThreshStratAllStrategic} is a solution. Hence we have an equilibrium. 
\end{proof}
\textbf{Uniform[0,1] example} When $K=2$, the solution of Equation \eqref{eq:keyEqThreshStratAllStrategic} and hence the equilibrium is obtained at $r=3/4$. For $K=3$, $r=2/3$; $K=4$ gives $r=15/24=0.625$; $K=5$ gives $r=.6$.


With appropriate shading functions, the bidders can recover the utility they would get when the seller was not optimizing her mechanism to maximize her revenue. Nevertheless, this equilibrium can be weakly collusive since we restrict ourselves to the class of functions introduced in Definition \ref{dfn:thresholded_beta}. It is not obvious that the strategy exhibited in Theorem \ref{thm:revBuyerInThreshStrategy} is an equilibrium in a larger class of functions. However, as mentioned previously, from a practical standpoint, as of now there exists no other clear way to increase drastically bidders utility that is independent from a precise estimation of the competition. The fact that at symmetric equilibrium bidders recover the same utility as in a second price auction with no reserves arguably makes it  an even more natural class of bidding strategies to consider from the bidder standpoint.

\subsection{Equivalence of revenue}

\begin{theorem}\label{thm:revBuyerInThreshStrategy}
Suppose we are in a symmetric situation and all buyers use the symmetric optimal strategy described above. 

Then the revenue of the seller is the same as in a second price auction with no reserve price. The same is true of the revenue of the buyers. 	
	
\end{theorem}
Interestingly, the theorem shows that this shading strategy completely cancels the effect of the reserve price. This is a result akin to our result on the Myerson auction. 

\begin{proof}
The revenue of the seller per buyer is 
$$
\Exp{\vValue_X(X)F^{K-1}(X)\indicator{X\geq r^*}}, \text{ with }\mathfrak{R}(r^*)=0\;.
$$	
Hence it is also 
$$
\Exp{\vValue_X(X)F^{K-1}(X)\indicator{X\geq r^*}}+\mathfrak{R}(r^*)=\Exp{\vValue_X(X)F^{K-1}(X)}\;.
$$
This is exactly the revenue of the seller in a second price auction with no reserve price. 

From the buyer standoint, his/her revenue is, since all buyers are using the same increasing strategy and the reserve value has been sent to 0, 
$$
\Exp{(X-\vValue_B(\beta(X;r^*)))F^{K-1}(X)}=\Exp{X F^{K-1}(X)\indicator{X\leq r^*}}+
\Exp{(X-\vValue_X(X)) F^{K-1}(X)\indicator{X> r^*}}\;.
$$
We know however that $\mathfrak{R}(r^*)=0$ and therefore
$$
\Exp{X F^{K-1}(X)\indicator{X\leq r^*}}=\Exp{(X-\vValue_X(X)) F^{K-1}(X)\indicator{X\leq r^*}}\;.
$$
Summing things up we get that his/her payoff
$$
\Exp{(X-\vValue_X(X)) F^{K-1}(X)}
$$
as in a second price auction with no reserve. 
\end{proof}
\section{Proofs for Subsections \ref{subsec:PhaseTransitionThresholding} and \ref{subsec:ExistenceNashEquilibriumIn2Stage}}
\label{sec:ProofsSec3Thresholding}
\input{ProofsSec3Thresholding.tex}

\section{Proof of results in Section 4.1}
\label{appendix_section5}
\subsection{Proof of Lemma \ref{lemma:impactOfMisestimationOnShading} and Corollary \ref{coro:impactOfMisestimationOnShadingThStrat}}
\begin{lemma}
\label{lemma:impactOfMisestimationOnShading} 
Suppose that the buyer uses a strategy $\beta$ under her value distribution $F$. Suppose the seller thinks that the value distribution of the buyer is $G$. 
Call $\lambda_F$ and $\lambda_G$ the hazard rate functions of the two distributions
Then the seller computes the virtual value function of the buyer under $G$, denoted $\psi_{B,G}$, as 
$$
\psi_{B,G}(\beta(x))=\psi_{B,F}(\beta(x))-\beta'(x)\left(\frac{1}{\lambda_G(x)}-\frac{1}{\lambda_F(x)}\right)\;.
$$	
\end{lemma}
\begin{corollary}\label{coro:impactOfMisestimationOnShadingThStrat}
If the buyer uses the strategy $\newThreshNotation(x)$ defined as
$$
\newThreshNotation(x)=\left(\frac{(r-\epsilon)(1-F(r))}{1-F(x)}+\epsilon\right)\indicator{x\leq r}+x\indicator{x>r}\;,
$$
we have, for $x\neq r$,
$\psi_{B,F}(\newThreshNotation(x))=\epsilon\indicator{x\leq r}+\psi_{F}(x)\indicator{x>r}\;.$
In particular, we have for $x\neq r$
$$
\left|\psi_{B,F}(\newThreshNotation(x))-\psi_{B,G}(\newThreshNotation(x)) \right|\leq \left|\psi_F(x)-\psi_G(x)\right|
 \left[(r-\epsilon)\indicator{x\leq r}+\indicator{x>r}\right]\;.
$$
If for all $x$, $\psi_{B,F}(\newThreshNotation)\geq \epsilon$ and $|\psi_F(x)-\psi_G(x)|\leq \delta$, we have 
$\psi_{B,G}(\newThreshNotation(x))\geq \epsilon -\delta\max((r-\epsilon),1)\;.
$
\end{corollary}
\begin{proof}
As we have seen before we have 
$$
\psi_{B,F}(\beta(x))=\beta(x)-\beta'(x)\frac{1-F(x)}{f(x)}\;.
$$
By construction, we have 
$$
\beta(x)-\beta'(x)\frac{1-F(x)}{f(x)}=0 \text{ for } x\leq r\;.
$$
If the seller perceives the behavior of the buyer under the distribution $G$, we have 
$$
\psi_{B,G}(\beta(x))=\beta(x)-\beta'(x)\frac{1-G(x)}{g(x)}\;.
$$
Hence, we have 
$$
\left|\psi_{B,G}(\beta(x))-\psi_{B,F}(\beta(x))\right|=\left|\beta'(x)\right|
\left|\frac{1-F(x)}{f(x)}-\frac{1-G(x)}{g(x)}\right|\;.
$$
Recall the hazard function $\lambda_F(x)=f(x)/(1-F(x))$. With this notation, we simply have 
$$
\left|\psi_{B,G}(\beta(x))-\psi_{B,F}(\beta(x))\right|=\left|\beta'(x)\right|\left|\frac{1}{\lambda_F(x)}-\frac{1}{\lambda_G(x)}\right|\;.
$$
The corollary follows by noting that when $x\leq r$, 
$$
\left|\beta_{r,\eps}'(x)\frac{1-F(x)}{f(x)}\right|\leq (r-\eps)\frac{1-F(r)}{1-F(x)}
$$
\end{proof}
Hence, a natural way to quantify the proximity of distributions in this context is of course in terms of their virtual value functions. Furthermore, if the buyer uses a shading function such that, under her strategy and with her value distribution, the perceived virtual value is positive, as long as the seller computes the virtual value using a nearby distribution, she will also perceive a positive virtual value and hence have no incentive to put a reserve price above the lowest bid.  In particular, if $\delta$ comes from an approximation error that the buyer can predict or measure, she can also adjust her $\epsilon$ so as to make sure that the seller perceives a positive virtual value for all $x$. 
\subsection{Proof of Theorem \ref{thm:epsAndMonopolyPriceByERM}}
\label{app:subsec:proofERMRobustness}
\begin{customthm}{5}
Suppose the buyer has a continuous and increasing value distribution $F$, supported on $[0,b]$, $b\leq \infty$, with the property that if $r\geq y\geq x$, $F(y)-F(x)\geq \gamma_F (y-x)$, where $\gamma_F>0$. Suppose finally that $\sup_{t\geq r}t(1-F(t))=r(1-F(r))$. 

Suppose the buyer uses the strategy $\newThreshNotation$ described above and samples $n$ values $\{x_i\}_{i=1}^n$ i.i.d according to the distribution $F$ and bids accordingly in second price auctions. Call $x_{(n)}=\max_{1\leq i \leq n}x_i$. In this case the (population) reserve value $x^*$ is equal to 0.

Assume that the seller uses empirical risk minimization to determine the monopoly price in a (lazy) second price auction, using these $n$ samples. 

Call $\hat{x}_n^*$ the reserve value determined by the seller using ERM. 

We have, if $C_n(\delta)=n^{-1/2}\sqrt{\log(2/\delta)/2}$ and $\epsilon>x_{(n)} C_n(\delta)/F(r)$ with probability at least $1-\delta_1$, 
$$
\hat{x}_n^*<\frac{2rC_n(\delta)}{\epsilon\gamma_F} \text{ with probability at least } 1-(\delta+\delta_1) \;.
$$

In particular, if $\epsilon$ is replaced by a sequence $\epsilon_n$ such that \\$n^{1/2}\epsilon_n min(1,1/x_{(n)}) \rightarrow \infty$ in probability, 
$\hat{x}_n^*$ goes to 0 in probability like $n^{-1/2}\max(1,x_{(n)})/\epsilon_n$. 
\end{customthm}
\textbf{Examples : } Our theorem applies for value distributions that are bounded,  with $\epsilon_n$ of order $n^{-1/2+\eta}$, $\eta>0$ and fixed. If the value distribution is log-normal$(\mu,\sigma)$ truncated away from 0 so all values are greater than a very small threshold $t$, standard results on the maximum of i.i.d $\mathcal{N}(\mu,\sigma)$ random variables guarantee that $x_{(n)}\leq \exp(\mu+\sigma \sqrt{2\log(n)})$ with probability going to 1. In that case too, picking $\epsilon_n$ of order $n^{-1/2+\eta}$, $\eta>0$ and fixed, guarantees that the reserve value computed by the seller by ERM will converge to the population reserve value, which is of course 0. 

\textbf{Comment : } The requirement on $\gamma_F$, which essentially means that the density $f$ is bounded away from 0 could also be weakened with more technical work to make this requirement hold only around 0, at least for the convergence in probability result. Similarly one could handle situations, like the log-normal case, where $\gamma_F$ is close to 0 at 0 by refining slightly the first part of the argument given in the proof. The formal proof is in Appendix \ref{appendix_section5}.
\begin{proof}

$\bullet$ \textbf{Preliminaries}\\
Notations~: We use the standard notation for order statistics $b_{(1)}\leq b_{(2)}\leq \ldots\leq b_{(n)}$ to denote our $n$ increasingly ordered bids. 
We denote as usual by $\hat{F}_n$ the empirical cumulative distribution function obtained from a sample of $n$ i.i.d observations drawn from a population distribution $F$. 

Setting the monopoly price by ERM amounts to finding, if $\hat{B}_n$ is the empirical cdf of the bids, 
$$
b^*_n=\argmax_t t (1-\hat{B}_n(t))
$$
We note in particular than 
$$
b_n^*\leq \max_{1\leq i\leq n}b_i=b_{(n)}\;,
$$
since $(1-\hat{B}_n(t))=0$ for $t> b_{(n)}$. 

Because $(1-\hat{B}_n(t))$ is piecewise constant and the function $t\mapsto t$ is increasing, on $[b_{(i)},b_{(i+1)})$ the function $t(1-\hat{B}_n(t))$ reaches its supremum at $b_{(i+1)}^-$. 
$$
b_n^*=\argmax_t t (1-\hat{B}_n(t))=\argmax_{1\leq i \leq n-1} b_{(i+1)}^- \left(1-\frac{i}{n}\right)\;.
$$
Since our shading function $\beta_{r,\eps}$ is increasing and if $x_{(i)}$ are our ordered values, we have, if $\hat{F}_n$ is the empirical cdf of our value distribution,  
$$
\argmax_{1\leq i \leq n-1} b_{(i+1)}^- \left(1-\frac{i}{n}\right)=\argmax_{1\leq i \leq n-1} \beta_{r,\eps}(u_{(i+1)}^-) \left(1-\frac{i}{n}\right)=\argmax_u \beta_{r,\eps}(u)(1-\hat{F}_n(u))\;.
$$
The last equality comes again from the fact that $(1-\hat{F}_n(u))$ is piecewise constant and $\beta_{r,\eps}(u)$ is increasing. So in our analysis we can act as if the seller had perfect information of our shading function $\beta_{r,\eps}$.

In what follows we focus on reserve values and denote 
\begin{gather*}
\hat{x}_{r,n}^*=\argmax_{0\leq x\leq r} \beta_{r,\eps}(x)(1-\hat{F}_n(x))\;,
\hat{x}_n^*=\argmax_{0\leq x} \beta_{r,\eps}(x)(1-\hat{F}_n(x))\;,\\
x^*=\argmax \beta_{r,\eps}(x)(1-F(x))
\end{gather*}
The arguments we gave above imply that $\hat{x}_{r,n}^*\leq x_{(n)}$. We will otherwise study the continuous version of the problem. We also note that by construction, $x^*=0$, though we keep it in the proof as it makes it clearer. 

%

We recall one main result of \cite{MassartTightConstantDKWAoP90} on the Dvoretzky-Kiefer-Wolfowitz (DKW) inequality: if $C_n(\delta)=n^{-1/2}\sqrt{\log(2/\delta)/2}$, 
$$
P(\sup_x |\hat{F}_n(x)-F(x)|>C_n(\delta))\leq \delta\;.
$$
In what follows, we therefore assume that we have a uniform approximation 
$$
\forall x\;, \; |\hat{F}_n(x)-F(x)|\leq C_n(\delta)\;,
$$
since it holds with probability $1-\delta$. In what follows we write $C_n$ instead of $C_n(\delta)$ for the sake of clarity. Using the fact that $\beta_{r,\eps}$ is increasing, this immediately implies that with probability at least $1-\delta$, for any $c>0$
$$
\forall x \in [0,c]\;, \; |\beta_{r,\eps}(x)(1-\hat{F}_n(x))-\beta_{r,\eps}(x)(1-F(x))|\leq \beta_{r,\eps}(c) C_n\;.
$$

$\bullet$ {\sf$\hat{x}_{r,n}^*=\argmax_{y\leq r} \beta_{r,\eps}(y)(1-\hat{F}_n(y))$ cannot be too far from $\mathsf{x^*}$}\\
Now for our construction of $\beta_{r,\eps}(x)$, we have by construction that
$$
\frac{\partial }{\partial u}\left[\beta(u) (1-F(u))\right]=-\eps f(u) \text{ when } x\leq r\;.
$$
In particular, it means that when $x,y\leq r$
$$
\beta_{r,\eps}(x)(1-F(x))-\beta_{r,\eps}(y)(1-F(y))=-\eps (F(x)-F(y))\;.
$$
Also $x^*=0$ since $\beta_{r,\eps}(1-F)$ is decreasing on $[0,r]$, as we have just seen that its derivative is negative.  Here we used the fact that $F$ is increasing. 

If $r>y\geq x^*+tC_n/\eps$, we have, using the previous inequality and the fact that  $\beta_{r,\eps}(1-F)$ is decreasing on $[0,r]$,
$$
\beta_{r,\eps}(y)(1-F(y))\leq \beta_{r,\eps}(x^*)(1-F(x^*))-\eps (F(x^*+tC_n/\eps)-F(x^*))\;.
$$
Since we assumed that $F(y)-F(x)\geq \gamma_F (y-x)$, we have 
$$
-\eps (F(x^*+tC_n/\eps)-F(x^*))\leq -tC_n \gamma_F\;.
$$
Since $\beta_{r,\eps}$ is increasing, we have $\sup_{0\leq x\leq r}\beta_{r,\eps}(x)\leq \beta_{r,\eps}(r)=r$. 
Picking $t>2r/\gamma_F$, it is clear that for $r>y\geq x+tC_n/\eps$, 
\begin{gather*}
\max_{r\geq u\geq x+tC_n/\eps} \beta_{r,\eps}(u)(1-\hat{F}_n(u))\leq 
\max_{r\geq u\geq x+tC_n/\eps} \beta_{r,\eps}(u)(1-F(u))+ \max_{r\geq u\geq x+tC_n/\eps} \beta_{r,\eps}(u) C_n
\\
\leq 
\beta_{r,\eps}(x^*)(1-F(x^*))+(r-t\gamma_F)C_n< \beta_{r,\eps}(x^*)(1-F(x^*))-rC_n\leq \beta_{r,\eps}(x^*)(1-\hat{F}_n(x^*))\;.
\end{gather*}
We conclude that $\hat{x}_{r,n}^*$ cannot be greater than $x+2rC_n/(\eps\gamma_F)$ and therefore
$$
\hat{x}_{r,n}^*-\hat{x}<\frac{2rC_n}{\eps\gamma_F}. 
$$
$\bullet$ \textbf{Dealing with $\max_{y>r} \beta_{r,\eps}(y)(1-\hat{F}_n(y))$}\\
Recall that $\max_{x} \beta_{r,\eps}(x)(1-F(x))=\beta_{r,\eps}(0)(1-F(0))=r(1-F(r))+\eps F(r)$. 
We now assume that $\max_{y\geq r} y(1-F(y))=r(1-F(r))$. This is in particular the case for regular distributions, which are commonly assumed in auction theory. 

To show that the argmax cannot be in $[r,b]$ with pre-specified probability we simply show that the estimated value of the seller revenue at reserve value 0 is higher than $\max_{y>r} \beta_{r,\eps}(y)(1-\hat{F}_n(y))$. Of course, 
$$
\max_{y>r} \beta_{r,\eps}(y)(1-\hat{F}_n(y))=\max_{x_{(n)}\geq y>r} \beta_{r,\eps}(y)(1-\hat{F}_n(y))\;.
$$

Recall that $\beta_{r,\eps}(0)=r(1-F(r))+\eps F(r)$. Under our assumptions, we have 
\begin{gather*}
\beta_{r,\eps}(0)(1-\hat{F}_n(0))=\beta_{r,\eps}(0)=r(1-F(r))+\eps F(r) \text{ and }\\
\max_{x_{(n)}\geq y\geq r} \beta_{r,\eps}(y)(1-\hat{F}_n(y))\\
\leq \max_{x_{(n)}\geq y\geq r} \beta_{r,\eps}(y)(1-F(y))+C_n \max_{x_{(n)}\geq y\geq r}\beta_{r,\eps}(y)\leq r(1-F(r))+C_n x_{(n)}
\end{gather*}
So as long as $\eps> x_{(n)}C_n/F(r)$, the result we seek holds. By assumption this property holds with probability $1-\delta_1$. 

The statement of the theorem holds when both parts of the proof hold. Since they hold with probability at least $1-\delta$ and $1-\delta_1$ the intersection event holds with probability at least $1-\delta-\delta_1$, as announced.  
\end{proof}

$\bullet$ \textbf{Asymptotic statement/Convergence in probability issue}\\
This is a straightforward application of the previous result and we give no further details.
\section{Proof of results of Section 4.2}
\label{proof_Section6}
This theorem works for non-regular value distributions and in the asymmetric case when the bidders have different value distributions.
\begin{customthm}{6}[Thresholding at the monopoly price]
\label{thresholding_monopoly_price}
Consider the one-strategic setting in a lazy second price auction with $F_{X_i}$ the value distribution of the strategic bidder $i$ with a seller computing the reserve prices to maximize her revenue.  Suppose $\beta_r$ is an increasing strategy with associated reserve value $r>0$. Suppose $\alpha = 0$, i.e. there is only an exploitation phase. 
Then there exists $\tilde{\beta}_r$  which does not depend on G such that  : 
\textbf{2)} $U_i(\tilde{\beta}_r)\geq U_i(\beta_r)$, $U_i$ being the utility of the strategic bidder.
\textbf{3)} $R_i(\tilde{\beta}_r)\geq R_i(\beta_r)$, $R_i$ being the payment of bidder $i$ to the seller. 
The following continuous functions fulfill these conditions for $\epsilon\geq 0$ small enough:
$$
\tilde{\beta}^{(\epsilon)}_r(x)=\left(\frac{[\beta_r(r)-\epsilon](1-F_{X_i}(r))}{1-F_{X_i}(x)}+\epsilon\right)\textbf{1}_{x<r}+\beta_r(x)\textbf{1}_{x\geq r}
$$
In the two-stage game, there is a critical value $\alpha_c$ for which, if $\alpha<\alpha_c$, it is preferable for the bidder to threshold at the monopoly price and if $\alpha>\alpha_c$, it is preferable for the bidder to bid truthfully. 	
\end{customthm}
\begin{proof}[Proof: Important special case of $\beta_r(x)=x$;]
We assumed that the seller computes the reserve price to maximize her revenue. Hence $r = \argmax x(1-F_{X_i}(x))$. In this case, $\int_{0}^{r} \psi_{i}(x_i)G(x_i)f_i(x_i)dx_i \leq 0$ (otherwise the reserve price would be lower and since the payment of bidder $i$ is equal to $\int_{r}^{+\infty} \psi_{i}(x_i)G(x_i)f_i(x_i)dx_i $. 

$\tilde{\beta}_r^{(\epsilon)}$ defined in Theorem \ref{thm:settingReserveValToZeroImprovesPerformance} verifies the ODE defined in Lemma \ref{lemma:keyODEs} such that $h_{\tilde{\beta}_r^{(\epsilon)}}(x) = \psi_{B_{\tilde{\beta}_r^{(\epsilon)}}}(\tilde{\beta_r}^{(\epsilon)}(x))=\epsilon$ for $x \in [0,r]$ and $h_{\tilde{\beta}^{(\epsilon)}}(x) = \psi_{X_i}(x)$ for $x \in [r,+\infty]$. $\tilde{\beta}_r^{(\epsilon)}$ is trivially increasing. 

Hence, the virtual value of the distribution induced by $\tilde{\beta}_r^{(\epsilon)}$ is non-negative on $[0,r]$ and the new reserve value is equal to zero. The new reserve price is therefore equal to the minimum bid of bidder $i$ and
\begin{equation*}
R_i(\tilde{\beta}_r^{(\epsilon)})\ =  R_i(\beta_r) + \mathbb{E}_{X_i \sim F_{i}}\bigg(\epsilon G_i(\tilde{\beta}_r^{(\epsilon)}(X_i))\textbf{1}(X_i \leq r)\bigg)\geq R_i(\beta_r)\;.
\end{equation*}
 The new bidder's utility is 
\begin{equation*} 
U_i(\tilde{\beta}_r^{(\epsilon)}) = U_i(\beta_r) + \mathbb{E}_{X_i \sim F_{i}}\bigg((X_i-\epsilon)G_i(\tilde{\beta}_r^{(\epsilon)}(X_i))\textbf{1}(X_i \leq r)\bigg)
\end{equation*}
For $\epsilon = 0$ we have clearly $U(\tilde{\beta}_r^{(\epsilon)})  \geq U(\beta_r)$ and $R_i(\tilde{\beta}_r^{(\epsilon)})\ =  R_i(\beta_r)$. Outside pathological cases, it is a strict inequality and by continuity with respect  to $\eps$, i.e. assuming $G$ continuous, it is true in a neighborhood of zero so for some $\eps>0$.
\end{proof}
We now handle rigorously the general case:
\begin{proof}
The reserve value $r>0$ is given. 
Consider 
$$
\tilde{\beta}_r(x)=
\begin{cases}
t_r(x)& \text{ if } x\leq r\\
\beta_r(x) & 	\text{ if } x > r
\end{cases}
$$
To make things simple we require $t_r(r)=\beta_r(r^+)$, so we have continuity. 
Note that beyond $r$ the seller revenue is unaffected. 
If the seller sets the reserve value at $r_0$ the extra benefit compared to setting it at $r$ is 
$$
\Exp{\vValue_{t_r}(t_r(X)G(t_r(X))\indicator{r_0\leq X <r}}\;.
$$
Hence, as long as $\vValue_{t_r}(x)\geq 0$, the seller has an incentive to lower the reserve value. The extra gain to the buyer is 
$$
\Exp{(X-\vValue_{t_r}(t_r(X)))G(t_r(X))\indicator{r_0\leq X <r}}\;.
$$
Now, if we take 
$$
t_r(x)=\frac{t_r(0)}{1-F(x)}\;,
$$
it is easy to verify that 
$$
\vValue_{t_r}(t_r(x))=t_r(x)-t_r'(x) \frac{1-F(x)}{f(x)}=0\;.
$$
So in this limit case, there is no change in buyer's payment and when the reserve price is moved by the seller to any value on $[0,r]$. If we assume that the seller is welfare benevolent, she will set the reserve value to 0.  To have continuity of the bid function, we just require that 
$$
\frac{t_r(0)}{1-F(r)}=\beta_r(r^+)\;.
$$
Since there is no extra cost for the buyer, it is clear that his/her payoff is increased with this strategy. 
Taking $t_r^{(\eps)}$ such that 
$$
\vValue_{t_r^{(\eps)}}(t_r^{(\eps)}(x))=\eps\;,
$$
gives a strict incentive to the seller to move the reserve value to 0, (so the assumption that s/he is welfare benevolent is not required) even if it is slightly suboptimal for the buyer. 
Note that we explained in Lemma \ref{lemma:keyODEs} how to construct such a $\vValue$. In particular, 
$$
t_r^{\eps}=\frac{C_\eps}{1-F(x)}+\eps, \text{ with } \frac{C_\eps}{1-F(r)}+\eps=\beta_r(r^+)
$$
works.
Taking limits proves the result, i.e. for $\eps$ small enough the Lemma holds, since everything is continuous in $\eps$. 

To extend to the two stage process, we can use exactly the same proof as in Lemma \ref{lemma:phaseTransition}.
\end{proof}
\textbf{Comment} We note that the flexibility afforded by $\eps$ is two-fold: when $\eps>0$, the extra seller revenue is a strictly decreasing function of the reserve price; hence even if for some reason reserve price movements are required to be small, the seller will have an incentive to make such move. The other reason is more related to estimation issues: if the reserve price is determined by empirical risk minimization, and hence affected by even small sampling noise, having $\eps$ big enough will guarantee that the mean extra gain of the seller will be above this sampling noise. Of course, the average cost for the bidder can be interpreted to just be $\eps$ at each value under the current reserve price and hence may not be a too hefty price to bear. 

For the sake of clarity, we first consider the case where the strategic bidder is only optimizing her utility in the second stage. 
We consider the robust-optimization framework \cite{aghassi2006robust}. We show that thresholding at the monopoly price is the best-response in the worst case of the competition when the number of players is not known to the strategic bidder. 

We assume that the bidders know that they all have the same value distribution but they do not know the number of bidders K. The problem is now to find the optimal thresholding parameter in the worst case : 
\begin{equation*}
r^{*} = \argmax_{r} min_{K \in \mathbb{N}} U(\beta_r,K)
\end{equation*}
We call this strategy the best robust-optimization strategy.

\begin{theorem}
\label{thresholded_monopoly_price_alpha=0}
Consider the one-strategic setting in a lazy second price auction with $F_{X_i}$ regular. Assume the symmetric setting where all the other bidders have the same value distributions and assume the non-strategic bidders are bidding truthfully. Consider the setting where $\alpha =1$. In the class of thresholded  strategies, the best robust-optimization strategy is to threshold at the monopoly price and bid truthfully after. 
\end{theorem}
\begin{proof} 
Based on Lemma \ref{lemma:directionalDerivativesThresholdedStrategiesRight}, we know that given a competition distribution $G$, the best response in the class of thresholded strategies is to bid truthfully above a certain threshold r. 

Given Lemma \ref{thresholding_monopoly_price}, we know that for any given competition distribution G, the optimal threshold r verifies $r\geq \psi^{-1}(0)$. 

For any $r >   \psi^{-1}(0)$, we show that there exists $K_{lim}$ such that $\forall K \geq K_{lim}$, it is better to bid truthfully on $[a,r]$ with $\psi^{-1}(0)<a<r$.
\begin{equation*}
U(\beta_{tr}) - U(\beta_r) = \int_{\psi^{-1}(0)}^r (x - \psi(x))F(x)^{K-1}f(x)dx - \int_{\psi^{-1}(0)}^r x \bigg(F\bigg(\frac{(1-F(r))r}{1-F(x)}\bigg)\bigg)^{K-1}f(x)dx
\end{equation*}
Since F is regular and by definition of the monopoly price, 
$$\forall x \in [\psi^{-1}(0),r], \frac{(1-F(r))r}{1-F(x)} \leq x$$
(In the case of equality $\psi(x)=0$). 
By definition of the virtual value, $\frac{x-\psi(x)}{x} \in )0,1($. Hence, if F is increasing (defined on all its support), for all $r > \psi^{-1}(0)$ there exists $a \in )\psi^{-1}(0),r($and $K_{lim}$ such that $\forall K \geq K_{lim}$, and 
$$\forall x \in [a,r], \frac{x - \psi(x)}{x} \geq \frac{\bigg(F\bigg(\frac{(1-F(r))r}{1-F(x)}\bigg)\bigg)^{K-1}}{F(x)^{K-1}}.$$
Hence thresholding at the monopoly price is optimal in the worst case of the competition distribution.
\end{proof} 
\section{Proof of results of Section 4.3}
\subsection{Thresholding at the monopoly price in the Myerson auction}
\label{appendix_myerson}
Here we ask what happens when one player is strategic, the others are truthful and she thresholds her virtual value at her monopoly price in the Myerson auction. 
\begin{customlemma}{3}
Consider the one-strategic setting. Assume all bidders have the same value distribution, $F_X$ and that $F_X$ is regular. Assume that bidder $i$ is strategic that the $K-1$ other bidders bid truthfully. Let us denote by $\beta_{truth}$ the truthful strategy and $\beta_{thresh}$ the thresholded strategy at the monopoly price. The utility of the truthful bidder in the Myerson auction $U_i^{Myerson}$ and in the lazy second price auction $U_i^{Lazy}$ satisfy
\begin{equation*} 
U_i^{\text{Myerson}}(\beta_{\text{thresh}}) - U_i^{\text{Myerson}}(\beta_{\text{truth}}) \geq U_i^{\text{Lazy}}(\beta_{\text{thresh}}) - U_i^{\text{Lazy}}(\beta_{\text{truth}})
\end{equation*}
\end{customlemma}
\begin{proof}
By definition the thresholded strategy at monopoly price $\beta_{\text{thresh}}$ verifies
$$
\psi_{B_i}(b_i)=\begin{cases} \psi_{X_i}(x_i) \text { if } x_i\geq \psi_{X_i}^{-1}(0)\\
\eps=0^+ \text{ if } x_i< \psi_{X_i}^{-1}(0)\;.
\end{cases}
$$
The revenue/utility of the buyer in the Myerson auction is 
$$
U_i^{\text{Myerson}}(\beta)=\Exp{(X_i-\psi_{B_i}(B_i))\indicator{\psi_{B_i}(B_i)\geq \max_{j\neq i}(0,\psi_{B_j}(B_j))}}\;.
$$
If we assume that all players except the i-th are truthful and the $i$-th player employs the strategy described above, we get that 
\begin{gather*}
U_i^{\text{Myerson}}(\beta_{\text{thresh}})=\Exp{(X_i-\psi_{B_i}(B_i))\indicator{\psi_{B_i}(B_i)\geq \max_{j\neq i}(0,\psi_{B_j}(B_j))}}\\
=\Exp{(X_i-\psi_{X_i}(X_i))\indicator{\psi_{X_i}(X_i)\geq \max_{j\neq i}(0,\psi_{X_j}(X_j))}\indicator{X_i\geq \psi_{X_i}^{-1}(0)}}\\
+\Exp{(X_i-\eps)\indicator{\eps\geq \max_{j\neq i}\psi_{X_j}(X_j)}\indicator{X_i<\psi_{X_i}^{-1}(0)}}
\end{gather*}
Note that here we've just split the integral into two according to whether $X_i$ was greater of less than the monopoly price and adjusted the definitions of $\psi_{B_i}(b_i)$ accordingly. 
We conclude immediately that 
\begin{gather*}
U_i^{\text{Myerson}}(\beta_{\text{thresh}})=U_i^{\text{Myerson}}(\beta_{\text{truth}})
+\Exp{(X_i-\eps)\indicator{\eps\geq \max_{j\neq i}\psi_{X_j}(X_j)}\indicator{X_i<\psi_{X_i}^{-1}(0)}}\;.
\end{gather*}
Therefore the extra utility derived from our shading is just 
$$
U_i^{\text{Myerson}}(\beta_{\text{thresh}}) - U_i^{\text{Myerson}}(\beta_{\text{truth}}) = \Exp{(X_i-\eps)\indicator{\eps\geq \max_{j\neq i}\psi_{X_j}(X_j)}\indicator{X_i<\psi_{X_i}^{-1}(0)}}
$$
If $X_j$'s are independent of $X_i$, we have 
$$
U_i^{\text{Myerson}}(\beta_{\text{thresh}})(\epsilon) - U_i^{\text{Myerson}}(\beta_{\text{truth}})(\epsilon) = \Exp{(X_i-\eps)P(\eps\geq \max_{j\neq i}\psi_{X_j}(X_j))\indicator{X_i<\psi_{X_i}^{-1}(0)}}
$$
When the $X_j$'s are independent of each other and we pick $\eps=0$ we get 
$$
\lim_{\eps\tendsto 0}U_i^{\text{Myerson}}(\beta_{\text{thresh}})(\epsilon) - U_i^{\text{Myerson}}(\beta_{\text{truth}})(\epsilon)=\left[\prod_{j\neq i}P(X_j\leq \psi_j^{-1}(0))\right]\Exp{X_i\indicator{X_i<\psi_{X_i}^{-1}(0)}}
$$
Going back to our work on second price auctions, the gain from thresholding in a continuous manner at 0 was 
$$
U_i^{\text{Lazy}}(\beta_{\text{thresh}}) - U_i^{\text{Lazy}}(\beta_{\text{truth}})=\Exp{XG(t_0(X))\indicator{X\leq \psi^{-1}(0)}}\;.
$$
In this case, $t_0(x)=\psi^{-1}(0) (1-F(\psi^{-1}(0)))/(1-F(x))$, $G(x)=F^{K-1}(x)$ using symmetry and independence. 
As we've seen above, 
$$
U_i^{\text{Myerson}}(\beta_{\text{thresh}}) - U_i^{\text{Myerson}}(\beta_{\text{truth}})=\Exp{X\indicator{X\leq \psi^{-1}(0)}}F^{K-1}(\psi^{-1}(0))=\Exp{X\indicator{X\leq \psi^{-1}(0)}}G(\psi^{-1}(0))\;.
$$
Now since $t_0(x)\leq \psi^{-1}(0)$, we see that in the symmetric case, for the 1-strategic player we have 
$$
\text{extra gain}_{\text{Myerson}}\geq \text{extra gain}_{\text{second price}}\;, 
$$
since $G(\psi^{-1}(0))\geq G(t_0(x))$ when $x\leq \psi^{-1}(0)$. 

Of course, in the symmetric case, the truthful revenue is the same in the Myerson and 2nd price lazy auction. As such the relative gain is also higher in the Myerson auction than in the second price auction.
\end{proof}
\paragraph{Example: $K$ symmetric bidders, unif[0,1] distribution, 1 strategic} In this case, $\psi_i^{-1}(0)=1/2$, $\int_0^{1/2}xf(x)dx=\int_0^{1/2}x dx=1/8$. And $\prod_{j\neq i}P(X_j\leq \psi_j^{-1}(0))=2^{-(K-1)}$. The truthful revenue/utility is the same as in a 2nd price auction with reserve at 1/2. Elementary computations show that this utility is $\int_{1/2}^1 x^{K-1}-x^K dx=(1-(K+2)/2^{K+1})/(K(K+1))$. 
\paragraph{Numerics for $K=2$} When $K=2$, the utility is 1/12 in the truthful case. And the extra utility is 1/16, i.e. $75\%$ of the utility in the truthful case.  Hence the utility of the shaded strategy is $7/48$. We note that the gain is even larger than for a 2nd price auction with monopoly reserve where the extra utility was $57\%$
\paragraph{Remarks about the asymmetric case} In the asymmetric case, it becomes harder to make a relative comparison of gains. That is because the truthful revenue in Myerson and 2nd price auctions are different. Furthemore, for the extra gain, even if $X_j$'s are independent, we have to compare 
$$
\Exp{X_i\indicator{X_i\leq \psi_i^{-1}(0)}\prod_{j\neq i}F_j(\psi_j^{-1}(0))} \text{ and }\Exp{X_i\indicator{X_i\leq \psi_i^{-1}(0)}\prod_{j\neq i}F_j\left(\frac{\psi_i^{-1}(0)(1-F_i(\psi_i^{-1}(0)))}{1-F_i(X_i)}\right)}
$$ 
In general, the comparison seem like it could go either way. An exception is the case $\psi_j^{-1}(0)\geq \psi_i^{-1}(0)$ for all $j\neq i$: then the extra revenue in the Myerson auction is greater than the extra revenue in the second price auction. 
 \subsection{Thresholding at the monopoly price in the eager second price auction with monopoly price}
 \label{appendix_eager}
\begin{customlemma}{4}
Consider the one-strategic setting with $F_{X_i}$ the value distribution of the strategic bidder $i$ that we assume regular. The K-1 other bidders have the same value distribution than bidder $i$ and bid truthfully. Lets not $\beta_{truth}$ the truthful strategy and $\beta_{thresh}$ the thresholded strategy at the monopoly price. The utility of the truthful bidder in the Myerson auction $U_i^{Myerson}$ and in the lazy second price auction $U_i^{Lazy}$ verifies
\begin{equation*} 
U_i^{Eager}(\beta_{thresh}) - U_i^{Eager}(\beta_{truth}) \geq U_i^{Lazy}(\beta_{thresh}) - U_i^{Lazy}(\beta_{truth})
\end{equation*}
\end{customlemma}
\begin{proof} 
The proof is very similar to the previous one. We recall that in the eager version of the second price auction, the winner of the auction is the bidder with the highest bid among bidders who clear their reserve price and she pays the maximum between the second highest bid and her reserve price.  

To complete the proof, we need to remark that in the symmetric case, the utility of the strategic bidder in the eager second price auction is 
$$
U_i^{\text{Eager}}(\beta)=\Exp{(X_i-\psi_{B_i}(B_i))G_i(\beta(X_i))\indicator{B_i\geq \max_{j\neq i}(psi_{B_j}^{-1}(0))}}\;.
$$
with $G_i$ the distribution of the highest bid of the competition above their reserve price. As we are in the symmetric case and all the bidders different than $i$ are bidding truthfully: 
$$
\forall x \in [0,\psi_X^{-1}(0)], G_i(x) = F_X^{K-1}(\psi_{X}^{-1}(0))\;.
$$
Hence, 
$$
U_i^{\text{Eager}}(\beta_{\text{thresh}}) - U_i^{\text{Lazy}}(\beta_{\text{lazy}}) = \mathbb{E}\bigg(X\big(F_X^{K-1}(\psi_{F_X}^{-1}(0)) - F_X^{K-1}(x)\big)\textbf{1}_{[X\leq \psi_{F_X}^{-1}(0)]}\bigg) \geq 0 
$$
\end{proof}

\end{document}

%% file: ProofsSec3Thresholding.tex
\paragraph{Proof of Theorem \ref{thm:TwoStepProcessCommitment}}
We recall the statement of the theorem.
\begin{theorem}
We call ${\mathcal G}_\alpha=\alpha G_1+(1-\alpha) G_2$. 
If the buyer uses the thresholding strategy $\tbeta$ of Definition \ref{dfn:thresholded_beta}  and commits to it, we have for their utility, if the seller is welfare benevolent~:
if $x_1\geq \psi^{-1}(0)$, 
\begin{equation}
U(x_1)=\Exp{X{\mathcal G}_\alpha\left(\frac{x_1(1-F(x_1))}{1-F(X)}\right)\indicator{X\leq x_1}}+\Exp{(X-\psi(X)){\mathcal G}_\alpha(X)\indicator{X\geq x_1}}\;.
\end{equation}
if $x_1<\psi^{-1}(0)$, 
\begin{equation}
U(x_1)=\alpha\Exp{X G_1\left(\frac{x_1(1-F(x_1))}{1-F(X)}\right)\indicator{X\leq x_1}}+\Exp{(X-\psi(X))\left[\alpha G_1(X)\indicator{X\geq x_1}+(1-\alpha) G_2(X)\indicator{X\geq \psi^{-1}(0)}\right]}\;.
\end{equation}

This utility has (in general) a discontinuity at $x_1=\psi^{-1}(0)$. For $x_1<\psi^{-1}(0)$, we have $U(0)\geq U(x_1)$. For $x_1>\psi^{-1}(0)$, the first order condition are the same as in \cite{nedelec2018thresholding,tang2016manipulate}, i.e. 
\begin{equation*}
{\mathcal G}_\alpha(r)=\Exp{\frac{X}{1-F(X)}g_\alpha\left(\frac{r(1-F(r))}{1-F(X)}\right)\indicator{X\leq r}}\;.
\end{equation*}
where the distribution of the competition is now ${\mathcal G}_\alpha=\alpha G_1+(1-\alpha) G_2$ and $g_\alpha$ is its density. Call $x_1^*(\alpha)$ the unique solution of this problem. 

Hence, the optimal threshold is $\argmax_{x_1\in \{0,x_1^*(\alpha)\}} U(x_1)$. An optimal threshold at $x_1=0$ corresponds to bidding truthfully. 

\end{theorem}

\begin{proof}[Proof of Theorem \ref{thm:TwoStepProcessCommitment}]

We first note that for the thresholding strategy at level $x_1$ we have
$$
\psi_B(\beta(x))=x\indicator{x\geq x_1} \text{ and } \beta(x)=\frac{x_1(1-F(x_1))}{1-F(x)}\indicator{x\leq x_1}+x\indicator{x>x_1}\;.
$$
Two situations are possible: if the bidder thresholds at $x_1>\psi^{-1}(0)$, then the welfare benevolent seller has an incentive to take all the bids in the second phase as their revenue curve is non-increasing in the value (because $\psi_B(\beta(x))\geq 0$ for all $x$). Hence the reserve price in the second phase is $x_1(1-F(x_1))$ and we have for the utility of the bidder in this case 
$$
U(x_1)=\Exp{X{\mathcal G}_\alpha\left(\frac{x_1(1-F(x_1))}{1-F(X)}\right)\indicator{X\leq x_1}}+\Exp{(X-\psi(X)){\mathcal G}_\alpha(X)\indicator{X\geq x_1}}\;,
$$
as announced in the text. 

On the other hand if $x_1<\psi^{-1}(0)$, the revenue of the seller is constant on $[0,x_1]$, increasing on $[x_1,\psi^{-1}(0)]$ and decreasing beyond $\psi^{-1}(0)$. Hence the optimal reserve price is $\psi^{-1}(0)$ regardless of $x_1<\psi^{-1}(0s)$. In this case we get 
$$
U(x_1)=\alpha\Exp{X G_1\left(\frac{x_1(1-F(x_1))}{1-F(X)}\right)\indicator{X\leq x_1}}+\Exp{(X-\psi(X))\left[\alpha G_1(X)\indicator{X\geq x_1}+(1-\alpha) G_2(X)\indicator{X\geq \psi^{-1}(0)}\right]}\;.
$$
The function $U(x_1)$ is clearly discontinuous at $\psi^{-1}(0)$. 

The fact that the optimal $x_1$ in $[0,\psi^{-1}(0))$ is 0 comes from the fact that the buyer derives no benefit from overbidding below $x_1$ in terms of reserve price in the second stage, so the auctions are effectively not optimized on past bids. Hence bidding truthfully is preferable in both stages. This corresponds to picking $x_1=0$. The claim about first order conditions follows from Appendix \ref{sec:proofNashAlphaZero}. Existence and uniqueness of $x_1^*(\alpha)$ follows from Lemma \ref{lemma:atMostOneNonTrivialSoln}.
\end{proof}

%

\paragraph{Proof of Lemma \ref{lemma:phaseTransition}}
\begin{proof}[Proof of Lemma \ref{lemma:phaseTransition}]
When $G_1=G_2=G$, we have ${\mathcal G}_\alpha=G$. Hence $U(r)$ does not depend on $\alpha$, when $r>\psi^{-1}(0)$.  Hence the optimal utility for the seller is independent of $\alpha$; let us call it $U^*(G)$.

On the other hand, we have 
$$
U(0)=\alpha\Exp{(X-\psi(X))G(X)}+(1-\alpha)\Exp{(X-\psi(X))G(X)\indicator{X\geq \psi^{-1}(0)}}\;.
$$
Obviously, $\Exp{(X-\psi(X))G(X)}> \Exp{(X-\psi(X))G(X)\indicator{X\geq \psi^{-1}(0)}}$. Hence $U(0)$ is increasing in $\alpha$. In the rest of the proof we denote $U(0)$ by $U(0;\alpha)$ to stress its dependence on $\alpha$. 

The results for $\alpha=0$ correspond to \cite{nedelec2018thresholding}.  So we know that as $\alpha\tendsto 0$ $U(0;\alpha)<U^*(G)$. Indeed, thresholding at the monopoly price $r=\psi^{-1}(0)$ beats bidding truthful in that case. And optimal thresholding beats thresholding at the monopoly price.  On the other hand, $U(0;\alpha=1)>U^*(G)$ since truthful bidding is optimal when there is no exploitation phase. Since $U(0;\alpha)$ is an increasing function of $\alpha$, and is obviously continuous, we see that there exists a unique $\alpha_c$ such that if $\alpha<\alpha_c$, it is preferable for the buyer to shade her bid and above it, it is preferable to bid truthful. 

\end{proof}

\paragraph{Proof of Theorem \ref{thm:ExistenceNashEqui2Stage}}
\begin{proof}[Proof of Theorem \ref{thm:ExistenceNashEqui2Stage}]
It is clear that each buyer's best response is either to bid truthful or to threshold optimally in the class of strategies we consider. 	

The theorem follows from Lemma \ref{lemma:phaseTransition}.	\textbf{1) Part 1.} For the first part, we suppose that all players except 1 plays the optimal thresholded strategy characterized by $r^*$ the unique solution of the equation 
$$
\frac{K-1}{r(1-F(r))}\Exp{XF^{K-2}(X)(1-F(X))\indicator{X\leq r}}=F^{K-1}(r)\;.
$$
This gives rise to the competition distribution 
$$
G_\beta(x;r^*)=F^{K-1}(\beta^{-1}(x))=\left(1-\frac{r^*(1-F(r^*))}{x}\right)^{K-1}\indicator{r^*(1-F(r^*))\leq x \leq r^*}+F^{K-1}(x)\indicator{x\geq r^*}\;.
$$
According to Lemma \ref{lemma:phaseTransition}, for this distribution $G_\beta(x;r^*)$, there exists $\alpha_{c,\text{thresh}}$ such that if $\alpha<\alpha_{c,\text{thresh}}$, the best response of the buyer in the class we consider is to threshold at $r>\psi^{-1}(0)$. Our results on Nash equilibria in the case $\alpha=0$ apply. By construction the optimal $r$ for this buyer is $r^*$ and we have the Nash equilibrium. 

Of course, if $\alpha>\alpha_{c,\text{thresh}}$, the best response of the buyer is to bid truthfully and hence there cannot be a Nash equilibrium in thresholding above $\psi^{-1}(0)$. 

\textbf{2) Part 2.} If all the players except one bid truthfully, Lemma \ref{lemma:phaseTransition} applies with $G=F^{K-1}$ and says that there exists $\alpha_{c,\text{truthful}}$, such that if $\alpha>\alpha_{c,\text{truthful}}$ the best response of the last player is to bid truthful. Hence there is a Nash equilibrium in the class we consider. On the other hand, if $\alpha<\alpha_{c,\text{truthful}}$, the best response of the last player is to threshold above $\psi^{-1}(0)$ and hence there is no Nash equilibrium.

\end{proof}

\subsection{Illustration of Theorem \ref{thm:TwoStepProcessCommitment}}

Figure \ref{fig:2StageNoReservesCommitment} on p. \pageref{fig:2StageNoReservesCommitment} provides an illustration of Theorem \ref{thm:TwoStepProcessCommitment}: for three different $\alpha$'s we plot the utility of the seller. 
\begin{figure}[t!]
    \centering
    \begin{subfigure}[t]{0.3\textwidth}
        \centering
        \includegraphics[width=\textwidth]{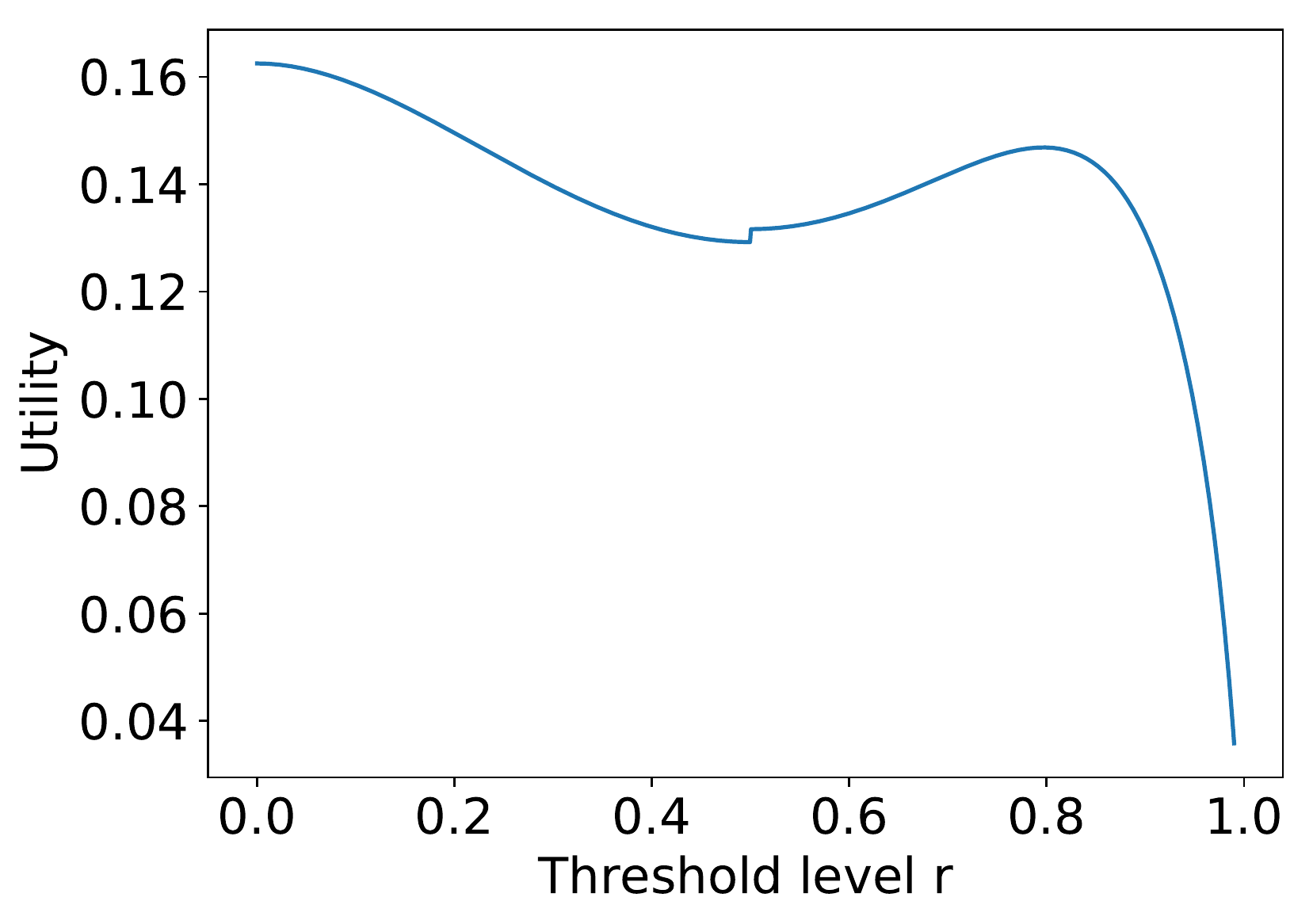}
        \caption{$\alpha=.95$}
    \end{subfigure}%
    ~
    \begin{subfigure}[t]{0.3\textwidth}
        \centering
        \includegraphics[width=\textwidth]{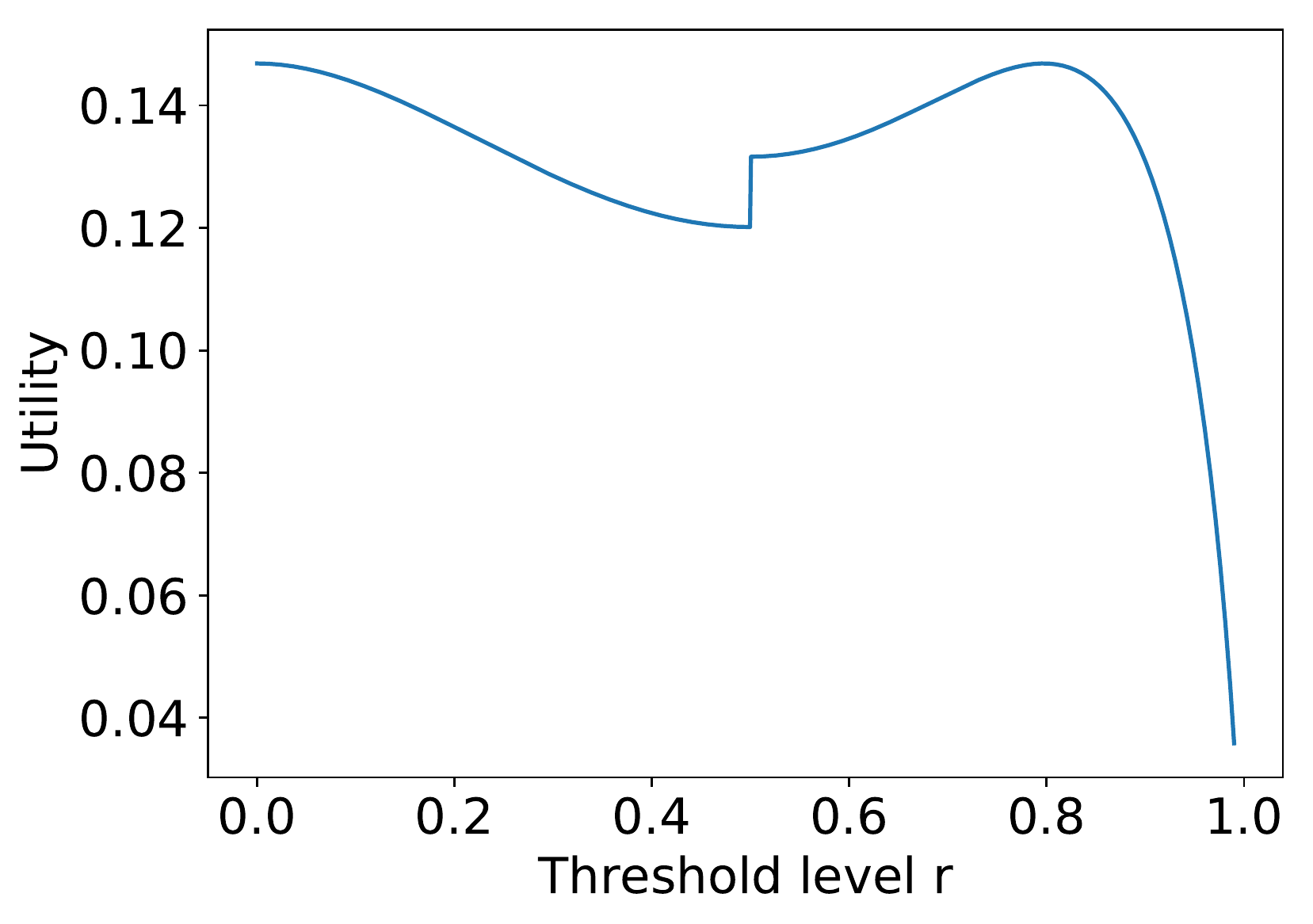}
        \caption{$\alpha=.76$}
    \end{subfigure}
	    ~
    \begin{subfigure}[t]{0.3\textwidth}
        \centering
        \includegraphics[width=\textwidth]{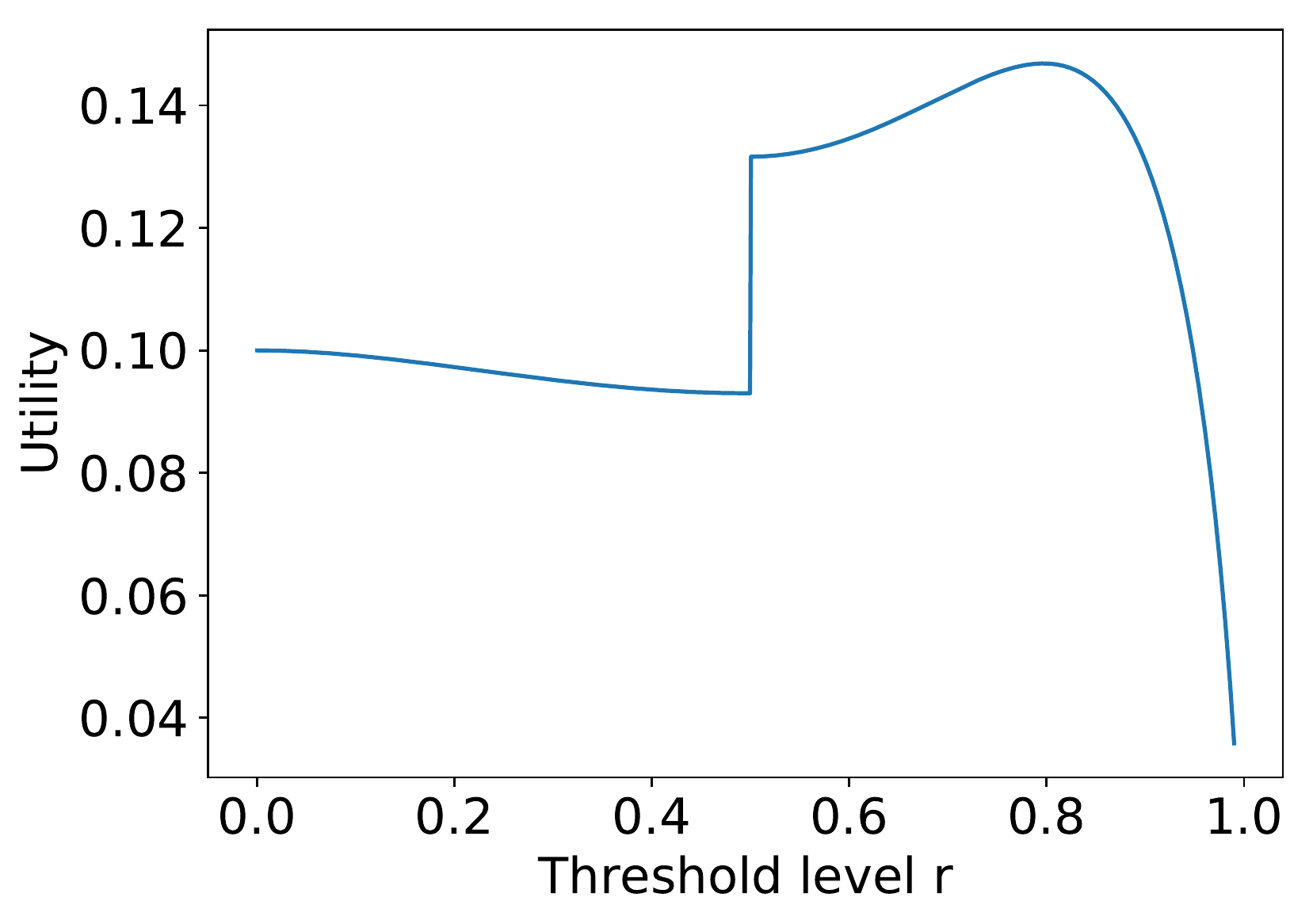}
        \caption{$\alpha=.2$}
    \end{subfigure}
    \caption{\textbf{Utility of the strategic buyer in 2 stage game for various threshold levels $\mathbf{r}$'s, zero reserves}. We consider a setup where we have two bidders, both have Unif[0,1] value distribution. One is strategic.  There is no reserve price in the first stage. Different $\alpha$'s are displayed.  At $\alpha\simeq .762$, the truthful strategy is essentially equivalent to optimal thresholding. For smaller $\alpha$'s, the optimal thresholding is preferable (right). For higher $\alpha$, truthful bidding is preferred. This illustrates the results of Theorem \ref{thm:TwoStepProcessCommitment}}  \label{fig:2StageNoReservesCommitment}
\end{figure}